\documentclass[journal,twoside,web]{ieeecolor}
\pdfoutput=1
\usepackage{generic}
\usepackage{cite}

\let\labelindent\relax
\usepackage{amsmath,amssymb,amsfonts,amsthm}
\usepackage{subcaption}
\usepackage{algorithmic}
\usepackage{graphicx}
\usepackage[outdir=./]{epstopdf}
\usepackage{textcomp}
\usepackage[ruled]{algorithm2e}
\SetKw{Break}{break}
\SetKw{Continue}{continue}
\usepackage{sample_complexity_hanchmin/hanch_sty}
\usepackage{cuted}
\usepackage{enumitem}
\setitemize{noitemsep,topsep=0pt,parsep=0pt,partopsep=0pt}
\usepackage{dsfont}
\newtheorem{definition}{Definition}
\newtheorem{remark}{Remark}
\newtheorem{lemma}{Lemma}
\newtheorem{claim}{Claim}
\newtheorem{theorem}{Theorem}
\newtheorem{corollary}{Corollary}
\newtheorem{proposition}{Proposition}
\newtheorem{assumption}{Assumption}

 \newtheorem{innercustomlem}{Lemma}
\newenvironment{customlem}[1]
  {\renewcommand\theinnercustomlem{#1}\innercustomlem}
  {\endinnercustomlem}
  
\renewcommand{\baselinestretch}{.985}

 \allowdisplaybreaks
  
\usepackage{ifthen}
\newboolean{showcomments}
\setboolean{showcomments}{true}
\usepackage{todonotes}

\newboolean{arxiv}
\setboolean{arxiv}{true}

\definecolor{bleudefrance}{rgb}{0.19, 0.55, 0.91}
\definecolor{ao(english)}{rgb}{0.0, 0.5, 0.0}

\newcommand{\addcite}[0]{\ifthenelse{\boolean{showcomments}}
{\textcolor{purple}{(add cite(s)) }}{}}%

\newcommand{\addcites}[0]{\ifthenelse{\boolean{showcomments}}
{\textcolor{purple}{(add cite(s)) }}{}}%

\newcommand{\enrique}[1]{  \ifthenelse{\boolean{showcomments}}
{\todo[inline,color=bleudefrance]{Enrique: #1}}{}}

\newcommand{\emmargin}[1]{\ifthenelse{\boolean{showcomments}}{\marginpar{\color{bleudefrance}\tiny EM: #1}}{}}

\newcommand{\juan}[1]{  \ifthenelse{\boolean{showcomments}}
{\todo[inline,color=pink]{Juan: #1}}{}}

\newcommand{\hancheng}[1]{  \ifthenelse{\boolean{showcomments}}
{\todo[inline,color=orange]{Hancheng: #1}}{}}

\newcommand{\agustin}[1]{  \ifthenelse{\boolean{showcomments}}
{\todo[inline,color=purple!60!white]{Agustin: #1}}{}}

\newboolean{showedits}
\setboolean{showedits}{true}
\usepackage[markup=underlined]{changes}
\definechangesauthor[color=bleudefrance]{EM}
\newcommand{\aem}[1]{
\ifthenelse{\boolean{showedits}}
{\added[id=EM]{#1}}
{\!#1\hspace{-4.75pt}}
}
\newcommand{\repem}[2]{
\ifthenelse{\boolean{showedits}}
{\replaced[id=EM]{#1}{#2}}
{\!#1\hspace{-4.75pt}}
}
\newcommand{\dem}[1]{
\ifthenelse{\boolean{showedits}}
{\deleted[id=EM]{#1}}
{}
}

\newcommand{\hl}[1]{#1}

\def\BibTeX{{\rm B\kern-.05em{\sc i\kern-.025em b}\kern-.08em
    T\kern-.1667em\lower.7ex\hbox{E}\kern-.125emX}}
\begin{document}


\title{
Learning to Act Safely
with Limited Exposure and Almost Sure Certainty}



%
\author{Agustin Castellano, Hancheng Min, Juan Andres Bazerque,
and Enrique Mallada 
\thanks{}\thanks{}\thanks{}\thanks{}\thanks{}
\thanks{A preliminary version of Section \ref{sec:bandits} was first presented in \cite{castellano2021learning}. A more in-depth discussion of Section \ref{sec:relaxed-rl} can be found in \cite{l4dc}.}
\thanks{A. Castellano, H. Min and E. Mallada are with the Department of Electrical and Computer Engineering at Johns Hopkins University, Baltimore, MD 21218 USA (e-mail: \{acaste11,hanchmin,mallada\}@jhu.edu).
 }
\thanks{
J. A. Bazerque is with the Department of Electrical Engineering at University of Pittsburgh, PA, 15213 (e-mail: juanbazerque@pitt.edu).
}
\thanks{This work was supported by NSF through grants CAREER 1752362, CPS 2136324, and TRIPODS 1934979}
}

\maketitle



\begin{abstract}
This paper puts forward the concept that learning to take safe actions in unknown environments, even with probability one guarantees, can be achieved without the need for an unbounded number of exploratory trials.
This is indeed possible, provided that one is willing to navigate trade-offs between optimality, level of exposure to unsafe events, and the maximum detection time of unsafe actions. We illustrate this concept in two complementary settings.
We first focus on the canonical multi-armed bandit problem and study the intrinsic trade-offs of learning safety in the presence of uncertainty. 
Under mild assumptions on sufficient exploration, we provide an algorithm that provably detects all unsafe machines in an (expected) finite number of rounds. The analysis also unveils a trade-off between the number of rounds needed to secure the environment and the probability of discarding safe machines. 

 We then consider the problem of finding optimal policies for a Markov Decision Process (MDP) with almost sure constraints. 
 We show that the action-value function satisfies a barrier-based decomposition which allows for the identification of feasible policies independently of the reward process. 
Using this decomposition, we develop a Barrier-learning algorithm, that identifies such unsafe state-action pairs in a finite expected number of steps. Our analysis further highlights a trade-off between the time lag for the underlying MDP necessary to detect unsafe actions, and the level of exposure to unsafe events. 
Simulations corroborate our theoretical findings, further illustrating the aforementioned trade-offs, and suggesting that safety constraints can speed up the learning process.


\end{abstract}

\begin{IEEEkeywords}
Uncertain systems, randomized algorithms, Markov processes, iterative learning control, optimal control
\end{IEEEkeywords}


\section{Introduction}\label{sec:intro}


Motivated by the success of machine learning in achieving super human performance, e.g., in vision
~\cite{rawat2017deep}, speech~\cite{hinton2012deep},
 and video games~\cite{silver2016mastering},
there has been recent interest in developing learning-enabled technology that can implement highly complex actions for safety-critical autonomous systems, such as self-driving cars, robots, etc. However, without proper safety guarantees such systems will rarely be deployed. There is therefore the need to develop analysis tools and algorithms that can provide such guarantees during, and after, training.
Efforts to provide such guarantees can be broadly grouped in two lines of work with somehow complementary success. 

The first approach leverages model-based techniques, based on Lyapunov stability~\cite{haddad2011nonlinear} and robust control~\cite{zhou1998essentials}, to provide (worst-case) safety guarantees based on a nominal model and assumptions on uncertainty and disturbances
\cite{2017.Berkenkamp,perkins2002lyapunov,2015.Berkenkamp, marvi}. 
In such settings, safety is usually specified in terms of stability, robust stability, or the existence of some invariant or control invariant sets.
\hl{Moreover, due to the worst-case approach to uncertainty, these methods tend to suffer poor performance in settings where the uncertainty is large, or the system performance is highly sensitive to model-uncertainty.}


The second line of work, in which ours naturally lies   , seeks to provide safety guarantees in model-free settings by adding constraints to the learning problem~\cite{Garcia,2018.Dalal,2018.Wachi,2020.Xu,2019.Paternainyfa,dxb,2020.Wachi,2020.HasanzadeZonuzy}. In this way, safety specifications can be further extended, beyond typical control notions, at the expense of introducing uncertainty and risk in the safety guarantees. \hl{More precisely, constraints in this body of work are probabilistic, either in expected value, critical value-at-risk, or high-probability, which do not allow for settings with hard constraints that need to be satisfied with probability one (w.p.1). Moreover, the satisfaction guarantees for these constraints are also generally provided probabilistically, via (expected) regret, or in high probability.} 

\hl{Naturally, both methodologies have their advantages. On the one hand, model-based methods can impose probability one constraints, and can also guarantee their satisfaction w.p.1. On the other hand, model-free methods can handle problems with high degree of uncertainty, and are particularly suited for cases where constraints and objectives are difficult to formalize.} Our work aims for robust safety guarantees more similar to the first approach (i.e., in an almost-sure sense), while borrowing methodologies from the second one (i.e., in a model-free setting), in an attempt to ``get the best of both worlds.''


\subsection{Contributions}
\hl{One of the driving arguments of our work is that, due to the logical (safe/unsafe) nature of safety assessment,  finding safe actions is a problem fundamentally different (and easier to solve)  than  finding the best one. To that end, we model safety assessment by a \emph{damage signal} $D_t\in\{0,1\}$, with $1$ indicating an unsafe event. We develop safety-assessment algorithms that learn which actions (or state-action pairs) satisfy safety specifications---w.p.1---both on Multi-armed Bandits (MABs)~\cite{lattimore2020bandit} and Markov Decision Processes (MDPs)~\cite{sutton2018reinforcement}.}

\hl{The algorithms we develop for both MABs and MDPs learn the set of all feasible policies. In both cases we provide \emph{explicit finite bounds} on the (expected) time needed to learn this set. It is important to note, though, that since our approach is model-free, unsafe events during training are unavoidable. That being said, during training our algorithms also \emph{limit the number of constraint violations}. We refer to the latter as \emph{exposure} throughout the paper, and define it both for MABs and MDPs.} 
\hl{Under both settings we also provide {relaxed} formulations of the problems that allow for less restrictive constraints. Specifically, we consider scenarios that admit policies that endure some damage along the trajectory. } 

\hl{Early developments for MABs were first presented in~\cite{castellano2021learning}. New to this work are extensions of all the results to $\lambda$-soft strategies (Definition \ref{def:lambdasoft}), as well as an improved bound for Theorem \ref{thm:mab-assured-finite-time}. The framework for MDPs was first introduced in \cite{l4dc}, but with no theoretical proofs. This paper contains all the proofs that support our theory.}

\subsection{Related work}
\hl{\textbf{Safety in multi-armed bandits} can be posed as guaranteeing that with high probability the expected cost of pulling an arm lies below a certain threshold. It has been studied in the linear setting \cite{amani2019linear,moradipari2020stage} and when both the reward and cost functions follow Gaussian Processes \cite{amani2020regret}, yielding methods similar to the celebrated Upper Confidence Bound algorithm \cite{auer} that first explore and expand a safe set and then focus on controlling regret \cite{pacchiano2021stochastic}. In contrast, in Section II we consider safety for a binary approach, where the probability of an arm failing follows a Bernoulli distribution (of unknown parameter). Instead of expanding a  safe region known a priori, our algorithm sequentially eliminates unsafe arms.}

\hl{\textbf{Safety in control} is usually specified in terms of reaching a certain region of the state space. In the context of learning, many works seek to apply Lyapunov design \cite{perkins2002lyapunov}, sometimes in tandem with Gaussian processes \cite{2017.Berkenkamp,2015.Berkenkamp} to obtain safe policies. Other proposed methods are related to using control barrier functions \cite{cbf1,cbf2} to specify safety guarantees and embed them into the cost to be optimized \cite{marvi}. While very powerful, these methods usually require knowledge of the system dynamics or a candidate Lyapunov/barrier function. Our framework on the other hand is model-free and thus assumes no prior knowledge. Safety specifications are implicitly captured by the damage signal and the transition kernel, yielding a safe region of the state space akin to what is achieved by computing backward reachable sets \cite{backward-reachable-sets, backward-reachable-sets-2}.}

\hl{\textbf{Constraints in MDPs}  are typically formulated as expectation-based constraints \cite{1998.Altman} or as critical value-at-risk \cite{cvar-mdps}. If the transitions and rewards are known, an optimal stationary policy for the first kind can be found as the solution of a linear program \cite{altman1998constrained}. 
Under unknown dynamics, typical strategies rely on primal-dual methods \cite{2020.Xu, 2019.Paternainwed, 2019.Paternainyfa,dxb, 2020.Ding, chen2016stochastic, triple-q}
, or on exploring and approximating the safety region \cite{2018.Wachi,2020.Wachi, 2018.Dalal}. The main difference in the approach of Section III is that we propose an MDP problem with almost sure constraints. 
Its particular nature proves useful in the sense that the set of feasible policies can be characterized in (expected) finite time.}

\hl{\textbf{Controller Synthesis} has also been used in the context of RL, where constraints are specified as high-level objectives given by Linear Temporal Logic \cite{hasanbeig2018,hasanbeig2019certified,hasanbeig2019reinforcement}. In these works the goal is  a policy that maximizes the probability of satisfying a set of tasks \cite{hasanbeig2019reinforcement}. 
These methods are powerful in the sense that they can capture a broad range of objectives and in some cases, even work on continuous spaces~\cite{lavaie-formal}. That being said, such methods rely on reward-shaping techniques that couple feasibility/safety with optimality. Our work decouples such problems (via the computation of a barrier function $B^*$). 
}

\subsection{Outline of the paper}
Section 
\ref{sec:multi_arm} addresses safety specifications in MABs, of the form $P(D_t=1|A_t)\leq \mu$ for some $\mu\geq 0$. 
While the \emph{``flawless''} case $\mu=0$ is rather straightforward (Section \ref{ssec:2.1}), the \emph{``relaxed''} case $\mu>0$ requires to trade-off between quickly discarding unsafe arms and accurately estimating $\prob(D_t=1|A_t)$, which we achieve with an arm-elimination algorithm based on Sequential Probability Ratio Tests (SPRTs) \cite{wald1945}. We focus on rapid (finite time) detection of unsafe actions almost surely, which naturally requires to (mildly) give up optimality by possibly discarding some safe arms (Section \ref{ssec:2.2}).

Section \ref{sec:assured_rl} deals with the problem of RL for  Constrained Markov Decision Processes.
 Dealing first with the \emph{flawless} setting, we develop a decomposition framework  (Section \ref{sec:decomposition}), based on hard barrier functions, that allows to decouple the safety assessment problem from the problem of maximizing rewards. This leads to a novel barrier learner algorithm, that is able to identify all state-action pairs that lead to unsafe events (Section \ref{ssec:3.2}). {Our analysis further shows the explicit role that the delayed consequences have in the learning process (Section \ref{ssec:3.3}).} \hl{Section \ref{sec:relaxed-rl} extends the RL setup to a relaxed setting in which we seek policies that allow for a finite number of unsafe events in any trajectory.}
 
 \hl{Numerical illustrations in Section \ref{sec:numericals} verify our theoretical analysis for MABs (Section \ref{sec:experiment-bandits}) and further suggest that, in the case of MDPs, learning the barrier first can aid in learning a task-oriented navigation goal later (Section \ref{sec:experiment-mdps}).}

\section{Multi-armed bandits}\label{sec:multi_arm}
\label{sec:bandits}
We consider the  setting of a stochastic bandit problem, with $K$ arms indexed as $a\in\{1,\ldots,K\}$. In the standard bandit problem an agent aims to devise an arm-pulling policy to optimize a  reward. Here, we switch focus to the safety problem, for which we consider that pulling some the arms could be unsafe and lead to system damage or harm to the agent. Specifically, at each round $t\geq 1$ the agent pulls an arm $A_t$ and obtains a binary-valued \emph{damage indicator} $D_t$. If $D_t$ is zero (one) this means that the action led to a safe (unsafe) result. We have that $\mathbb{E}[D_t | A_t]=\mu_{A_t}$. Each machine is therefore characterized by its \emph{safety parameter} $\mu_a$. The greater this value is, the more likely it is that pulling the machine will lead to an unsafe event. The goal of the player in this setup is to identify all the machines that are $\mu$-unsafe, which is hereby defined.

\begin{definition}
Given a \emph{safety specification} $\mu\in[0,1)$, a machine $a$ is said to be $\mu$-unsafe if and only if $\mu_a>\mu$. 
Accordingly, a machine is $\mu$-safe whenever $\mu_a\leq\mu$.
\end{definition}

The safety requirement $\mu$ is a design parameter, and is the only data that the player has access to along with the signal $D_t$. We will look at two distinct cases:
\begin{enumerate}[label=\Alph*)]
    \item \emph{Flawless setting} ($\mu=0$): in this setting any machine with positive probability of giving damage (i.e. $\mu_a>0$) is considered unsafe.
    \item \emph{Relaxed setting} ($\mu>0$): In this setting we want to identify unsafe machines with $\mu_a>\mu$. This means that we allow ``somewhat defective'' machines.
\end{enumerate}
We will focus first on the flawless setting, as it will allow us to build intuition on how to devise a proper Algorithm and on the values of metrics involved. The solution in this case is straightforward: let the agent pull each arm and avoid arms that have led to an unsafe event $D_t=1$. For the second case, we will rely on building a one-sided Sequential Probability Ratio Test (SPRT) \cite{wald1945}, that will make us try each machine a sufficient number of times; if the machine is \emph{unsafe}, the test will eventually decide on that hypothesis.

In both cases, our goal is the same. We want to \emph{detect} all the machines that are unsafe. To that end, let us define at each round $t\geq 1$ the candidate safe set $\mathcal{A}_t$, which contains all the arms that haven't been classified as unsafe. 

\begin{assumption}
Given a safety requirement $\mu$, there are $M$ $\mu$-unsafe machines,  where $1\leq M\leq K$. Without loss of generality we will assume the the first $M$ arms to be $\mu$-unsafe (i.e. $\mu_a>\mu, ~ a=1,\ldots,M$)
\end{assumption}

\begin{definition}
For each round $t\geq 1$ we define the \emph{candidate safe set} $\mathcal{A}_t$ as the set containing all the arms that have not been classified as unsafe.
\end{definition}

\hl{The set $\mathcal A_t$ is initialized as $\mathcal{A}_0:=[K]=\left\{1,\ldots K\right\}$, and sequentially trimmed down whenever an arm is found to be unsafe. To select which arm to pull at each round, we consider a \emph{strategy} $\psi$, which assigns to a set of arms  a probability of sampling each arm in the set. We define it next.}
\begin{definition}[Strategy]
\hl{A strategy $\psi:2^{[K]}\rightarrow \Delta^K$ is a function mapping sets of machines to the $K$-dimensional probability simplex, with the property that for any set $\mathcal{A}\subseteq[K]$ it holds that the support of $\psi(A)$ satisfies $\operatorname{Supp}\left(\psi\left(\mathcal{A}\right)\right)\subseteq\mathcal{A}.$}
\end{definition}
\hl{The requirement that $\operatorname{Supp}\left(\psi\left(\mathcal{A}\right)\right)\subseteq\mathcal{A}$ ensures that (for any set $\mathcal{A}$) the strategy only assigns positive mass to arms in $\mathcal{A}$. Then, given a set $\mathcal A$, the probability of sampling an arm $a$ is $\mathbb{P}_{A\sim\psi(\mathcal{A})}(A=a)$, where $A$ is the (random) variable corresponding to the arm being sampled. We will further use the notation $\mathbb{P}_\psi(A=a)$ when the set $\mathcal A$ is understood from the context.
Thus, given a set $\mathcal A_t$ at time $t\geq0$, $\psi$ induces a probability over a random action, $A_t$, to be taken with probability $\mathbb{P}_\psi(A_t=a)$.
}

\hl{We now define a class of strategies that are \emph{sufficiently exploratory}, in the sense that they sample each arm in the candidate set with positive probability.}


\begin{definition}[$\lambda$-soft strategy]\label{def:lambdasoft}
    {Given $0<\lambda\leq 1$, a strategy $\psi$ is called $\lambda$-soft if $\forall \mathcal{A}\subset [K]$ $$\mathbb{P}_{\psi}\left(A=a\right)\geq \frac{\lambda}{|\mathcal{A}|}\quad\forall a\in\mathcal{A}\,. $$}
\end{definition}
{As stated above, $\lambda$-softness is a condition on sufficient exploration of each arm. As a special case, the always-uniform strategy is $1$-soft.}

Although we recognize that detecting unsafe machines necessarily implies pulling from those unsafe arms, we want to have a notion of whether our decision rules choose machines in an efficient manner. It is with that goal in mind that we define at each round the \emph{exposure}.

\begin{definition}[Exposure]\label{def:exposure}
For $t\geq 1$ we define 
\begin{equation}
    E_t = \sum_{\tau=1}^t \mathds{1}\left(\mu_{A_\tau}>\mu\right).
\end{equation}
\end{definition}
\noindent where $\mathds{1}(x)=1$ if $x$ is true and 0 otherwise. This metric  counts the rounds in which  $\mu$-unsafe machines have been pulled, regardless of whether they led to an unsafe event or not. Notice that the exposure  inherits the randomness of the sequence of decisions $A_t$. Our results throughout this section will deal then with the expected value $\mathbb{E}[E_t]$. Ideally a good player would be one that attains low exposure---meaning it selects unsafe machines infrequently. 
\begin{remark}
\hl{Throughout the remainder of this section, we  show that for $\lambda$-soft strategies $\mathbb{E}[E_t]$ is bounded.} This marks a stark contrast with the notion of \emph{regret}, typically studied in Bandit settings~\cite{lattimore2020bandit}, where unbounded regret is unavoidable~\cite{1985.Lai}.
\end{remark}
At time $t$, the number of times arm $a$ has been pulled is:  
\begin{equation}
    N_a(t)=\sum_{\tau=1}^t \mathds{1}(A_\tau=a)\,.
\end{equation}
The following lemma states that the expected exposure coincides with the sum of the expected number of pulls over the unsafe machines.

\begin{lemma}\label{lemma:handicap_tnt}
$\mathbb{E}[{E_t}]=\sum_{a=1}^M\mathbb{E}[N_a(t)]$
\begin{proof}
\begin{align*}
    &\mathbb{E}[E_t] \!=\!  \sum_{\tau=1}^t \mathbb{E}\left[\mathds{1}\left(\mu_{A_\tau}\!>\!\mu\right)\right]\!=\!\sum_{\tau=1}^t\sum_{a=1}^K P(A_\tau\!=\!a)\mathds{1}\left(\mu_{a}\!>\!\mu\right)\\
    &
    \!=\!\sum_{\tau=1}^t\sum_{a=1}^M P(A_\tau\!=\!a)
    \!=\!\sum_{a=1}^M\sum_{\tau=1}^t \mathbb{E}\left[\mathds{1}(A_\tau\!=\!a)\right]\!=\!\sum_{a=1}^M\mathbb{E}[N_a(t)]
\end{align*}
\end{proof}
\end{lemma}

\begin{definition}
The \emph{conservation ratio} $C_{\varepsilon,t}$ is the proportion of safe machines kept at time $t$
\begin{equation}
C_{\varepsilon,t}:=\frac{|\mathcal{A}_t\cap\mathcal{A_\varepsilon^\star}|}{|\mathcal{A}_\varepsilon^\star|}\label{eq:safety_ratio}
\end{equation}
where $\mathcal{A}_t$ is the candidate safe set and $\mathcal{A}_\varepsilon^\star$ is the set containing all arms that are $\left(\mu-\varepsilon\right)$-safe:  $\mathcal{A}_\varepsilon^\star=\{a\in\mathcal{A}:\mu_a\leq\mu-\varepsilon\},$
where $0\leq \varepsilon\ll 1$ is a non-negative slack 
parameter. \hl{We let $C_{\varepsilon, t}\equiv 0$ in the event $|\mathcal{A}_\varepsilon^\star|=0$.}
\end{definition}
This ratio gives the proportion of $(\mu-\varepsilon)$-safe machines present in the candidate safe set at each time 
step. ($C_{\varepsilon,t}$ close to $1$ is desirable). The need for the conservativeness given by $\varepsilon$ is that we want to detect unsafe machines in finite time. This will become clearer when we discuss the \emph{relaxed setting}, for now it suffices to assume $\varepsilon=0$.
\subsection{Flawless setting ($\mu=0$)}\label{ssec:2.1}
We start with the simplest case imaginable, which is that of a rigorous safety requirement of $\mu=0$. In this setting, any machine that has positive probability of giving damage $D_t=1$ should be deemed unsafe. The strategy for discarding these machines is pretty straightforward: at each round $t$ select an arm $A_t$ following strategy $\psi(\mathcal{A}_t)$ and, if the resulting damage is $D_t=1$, classify the machine as unsafe by taking it out of the candidate safe set $\mathcal{A}_t$. This decision rule, summarized in Algorithm \ref{alg:flawless-inspector} has two interesting properties: \emph{i)} all unsafe machines are eventually found, and \emph{ii)} no safe machines are discarded along the way, as the the following theorem states. 
\begin{algorithm}[h]
    \SetAlgoLined
    \DontPrintSemicolon
    \KwIn{Number of arms $K$, strategy $\psi$.}
    \tcc{Initialize candidate safe set}
    $\mathcal{A}_0=\left\{1,\ldots,K\right\}$ 
    
     \For{$t=1,2,\ldots,$}{
    Pick arm $A_t \sim \psi(\mathcal{A}_t$)
    
    Observe damage $D_t$
    
    \If{$D_t=1$}{
    \tcc{trim unsafe arm} 
    $\mathcal{A}_t \leftarrow \mathcal{A}_{t-1}\setminus \{A_t\}$
    }
    }
    \caption{Flawless Inspector}
    \label{alg:flawless-inspector}
\end{algorithm}

\begin{theorem}\label{thm:mab-assured-finite-exposure}
Under Algorithm \ref{alg:flawless-inspector}, for every strategy $\psi$, the following (in)equalities hold with probability 1 for all $t\geq 1$:
\begin{align}
    \mathbb{E}[C_{0,t}] &= 1\\
    \mathbb{E}[E_t] &\leq \sum_{a=1}^M \frac{1}{\mu_a} \label{eq:handicap-bound}
\end{align}
\begin{proof}
The proof for the safety ratio $C_{0,t}$ is immediate, since Algorithm \ref{alg:flawless-inspector} can never discard a safe  machine (the event $D_t=1$ has zero probability when pulling from a flawless arm). For the remaining equalities, let $N_a$ be the number of pulls of the $a$-th arm needed to classify it as unsafe, which is well-defined \hl{(finite with probability one)} for $a=1,\ldots, M$. We have that $N_a\sim\mathrm{Geometric}(\mu_a)$, hence:
%
%
\be
    \mathbb{E}[N_a]=\sum_{n=1}^\infty \prob(N_a=n)n=\sum_{n=1}^\infty \mu_a(1-\mu_a)^{n-1}n=\frac{1}{\mu_a}\label{eq:Na}
\ee
Furthermore, for all $t$ we have $N_a(t)\leq N_a$. Taking expectation on both sides and using \eqref{eq:Na} yields $\mathbb{E}[N_a(t)]\leq\frac{1}{\mu_a}$. {Combining this with the result of Lemma \ref{lemma:handicap_tnt} gives \eqref{eq:handicap-bound}.}
\end{proof}
\end{theorem}
We now state that under a uniform strategy, Algorithm \ref{alg:flawless-inspector} finds all unsafe machines in expected finite time.

\begin{theorem}\label{thm:mab-assured-finite-time}
{Let $\psi$ be a $\lambda$-soft strategy, and assume all $M$ unsafe machines satisfy $\mu_a\geq\mu_{\text{low}}$. Then Algorithm \ref{alg:flawless-inspector} finds all unsafe machines in time $T$, whence:}
\begin{equation}
    {\mathbb{E}[T] \leq \frac{1}{\lambda\mu_{\text{low}}}\big(M + (K-M)\log(M+1)\big).}
\end{equation}
\end{theorem}
\begin{proof}
The proof is in Appendix \ref{app:mab-assured-unif}.
\end{proof}
\begin{corollary}[Sample-complexity bound]
\hl{For any $\delta\in(0,1)$, with probability at least $1-\delta$, Algorithm \ref{alg:flawless-inspector} finds all unsafe machines after at most}
$$
\hl{\left(1+\log\frac{1}{\delta}\right)\frac{1}{\lambda\mu_{\text{low}}}\big(M + (K-M)\log(M+1)\big).}
$$
\begin{proof}
    \hl{The stopping time $T$ defined in the previous theorem is a sum of geometric random variables (see Appendix \ref{app:mab-assured-unif}). We use the bound on its expected value along with a tail-bound for sums of geometric random variables \cite[Corollary 2.4]{janson2018tail} with mean $1+\log 1/\delta$.}
\end{proof}
\end{corollary}


\subsection{Relaxed setting ($\mu>0$)}\label{ssec:2.2}
We next consider the \emph{relaxed} setting, in which we allow machines that give damage $D_t$ with (possibly low) probability $\mu$. This means that in order to identify unsafe machines we can no longer discard them at first sign of damage, but must rather pull from each arm and observe multiple unsafe events in order to be confident that the machine in question is defective. To check whether an arm $a$ is defective or not we build the following hypothesis test $\mathbb{H}_a$:
\begin{equation}\label{eq:hyp-test}
    (\mathbb{H}_a)~\begin{cases}\mathcal{H}_0:&\mu_a\leq\mu-\varepsilon\\\mathcal{H}_1:&\mu_a>\mu\end{cases}
\end{equation}
in which the alternative hypothesis is that the machine is $\mu$-unsafe, and where we introduce the slack parameter $\varepsilon\in(0,\mu]$. In order to solve \eqref{eq:hyp-test}, we will devise a Sequential Probability Ratio Test (SPRT) which is based on Abraham Wald's seminal work \cite{wald1945} with the following properties:
\begin{enumerate}
    \item If the machine is unsafe ---meaning $\mathcal{H}_1$ is true--- the test will terminate in expected finite pulls $\mathbb{E}[N_a]$.
    \item If the machine is safe, the probability that the test  ---incorrectly--- decides on $\mathcal{H}_1$ is upper bounded by $\alpha$, where $\alpha\in (0,1)$ is the \emph{failure tolerance} of the test.
    \item Lowering failure tolerance $\alpha$ necessarily implies more pulls $N_a$ to detect unsafe machines.
    \item If $\mu_a\in(\mu-\varepsilon,\mu)$ the test is inconclusive.
    \item The test is one-sided: it only decides on $\mathcal{H}_1$. Similar to Algorithm \ref{alg:flawless-inspector}, whenever a machine is classified as unsafe it is not pulled any longer.
\end{enumerate}




For a fixed arm $a$, let $\mathbf{d}_a(t)=\{D_\tau: A_\tau=a, \tau\leq t\}$ be the (binary) sequence of outcomes of the $a$-th machine up to time $t$. The sequential probability ratio test relies on computing the log-likelihood ratio at each time step:

    \begin{equation}\label{eq:log-likelihood-def}
    \Lambda_a(t)=\log\frac{f_\mu(\mathbf{d}_a(t))}{f_{\mu-\varepsilon}(\mathbf{d}_a(t))}\,,
\end{equation}
where $f_\mu$ and $f_{\mu-\varepsilon}$ are the likelihood that the sequence $\mathbf{d}_a(t)$ came from independent Bernoulli trials of success rate $\mu$ and $\mu-\epsilon$ respectively. The test terminates by declaring $\mathcal{H}_1$ whenever
\begin{equation}\label{eq:log-likelihood-threshold}
    \Lambda_a(t) \geq \log(1/\alpha)\;.
\end{equation}     

By means of sufficient statistics, $\Lambda_a(T)$ can be written as a function of both $k$, the total number of outcomes of $D_t=1$ and $N_a(T)$, the total number of pulls up to time $T$. 
For a particular single arm the testing procedure is as follows.   For each round $t\geq 1$ pull the arm and (given $\mu$ and $\varepsilon$) update the log-likelihood in \eqref{eq:log-likelihood-def}. If  \eqref{eq:log-likelihood-threshold} holds, then terminate the test, otherwise observe another sample $D_t$ and repeat. 

The following two lemmas state the desired behavior of the SPRT. Namely, that \emph{i)} if the machine is unsafe, the SPRT will declare $\mathcal{H}_1$ with probability 1, \emph{ii)} if the machine is safe, the SPRT will (incorrectly) declare $\mathcal{H}_1$ with probability less than or equal to $\alpha$, and  \emph{iii)} the time of detection for unsafe machines is finite in expectation, and is well characterized in terms of the design parameters $\mu$, $\varepsilon$ and $\alpha$. 

\begin{lemma}\label{lemma:unsafe_sequences_escape}
For a fixed arm $a$ of parameter $\mu_a$, consider the sequential probability ratio test defined by \eqref{eq:hyp-test}--\eqref{eq:log-likelihood-threshold}, where $\mu$, $\varepsilon$ and $\alpha$ are given. Then:
\begin{enumerate}[label=\roman*)]
    \item If $\mathcal{H}_1$ is true, the test will (correctly) declare $\mathcal{H}_1$ with probability 1.
    \item If $\mathcal{H}_0$ is true, the test will  keep going indefinitely with probability $\geq 1-\alpha$
\end{enumerate}
\end{lemma}
\begin{proof} 
\ifthenelse{\boolean{arxiv}}{The proof is in the Appendix \ref{app:relaxed-lemmas}.}{The proof is in the Appendix B of \cite{self-cite}.}
\end{proof}

\begin{lemma}\label{lemma:sprt_detection_time} For a fixed arm $a$ of parameter $\mu_a$, consider the sequential probability ratio test defined by \eqref{eq:hyp-test}--\eqref{eq:log-likelihood-threshold}, where $\mu$, $\varepsilon$ and $\alpha$ are given. Then, if the alternative $\mathcal{H}_1$ is true, the test is expected to terminate after $T_a$ steps, with: \begin{equation}
    \mathbb{E}[T_a]\leq 1+\frac{\log\left(1/\alpha\right)}{\emph{kl}(\mu,\mu-\varepsilon)}\;,
\end{equation}
where $\emph{kl}(\mu,\mu-\varepsilon)$ is the Kullback-Leibler divergence between Bernoulli distributions:
\begin{equation}
\emph{kl}(\mu,\mu-\varepsilon)=\mu\log\frac{\mu}{\mu-\varepsilon}+\left(1-\mu\right)\log\frac{1-\mu}{1-\mu+\varepsilon}\, .
\label{eq:kl-div}\end{equation}
\end{lemma}
\begin{proof} 
\ifthenelse{\boolean{arxiv}}{The proof is in the Appendix \ref{app:relaxed-lemmas}.}{The proof is in the Appendix B of \cite{self-cite}.}
\end{proof}

\begin{remark}
The preceding lemma elucidates the need for the slack parameter $\varepsilon$: separating the two limiting distributions enables termination of the test under $\mathcal{H}_1$ in finite time. It also unveils two fundamental trade-offs. Firstly, if the distance between the limiting distributions increases (by enlarging $\varepsilon$), then the test detects unsafe machines faster. However, it becomes inconclusive over a larger proportion of the machines, since nothing can be assured in the region $\mu_a\in(\mu-\varepsilon,\mu)$. Secondly, $\alpha$ can be increased in order to detect unsafe machines faster, though this comes at the cost of declaring $\mathcal{H}_1$ for a larger proportion of the safe machines.
\end{remark}

We build Algorithm \ref{alg:relaxed_inspector} based on this SPRT, and state its main  properties in the following theorem.


\begin{algorithm}
\KwIn{Number of arms $K$, strategy $\psi$,\\ requirement $\mu$, slack $\varepsilon$, tolerance $\alpha$.}
    \tcc{Init. safe set and ratios}
    
$\mathcal{A}_0=\left\{1,\ldots,K\right\}$,\ $\Lambda_a=0~~\forall a=1,\ldots,K$\\
\For{$t=1,2\ldots$}{
Pick arm $A_t \sim \psi(\mathcal{A}_t)$\\
Observe damage $D_t$\\
\tcc{Update log-likelihood ratio}
$\Lambda_{A_t}\leftarrow\Lambda_{A_t}+\log\frac{f_\mu(D_t)}{f_{\mu-\varepsilon}(D_t)}$\\

\If{$\Lambda_{A_t}\geq \log(1/\alpha)$}{\tcc{SPRT ends, trim unsafe arm}$\mathcal{A}_t \leftarrow \mathcal{A}_{t-1}\setminus \{A_t\}$
}

}
\caption{Relaxed Inspector}
\label{alg:relaxed_inspector}
\end{algorithm}

\begin{theorem}\label{thm:mab-relaxed-finite-time}
Under Algorithm \ref{alg:relaxed_inspector}, for every strategy $psi$, the following inequalities hold with probability 1 for all $t$:
\begin{align}
\mathbb{E}[C_{\varepsilon,t}]&\geq 1-\alpha \label{eq:rho_bound}\\
\mathbb{E}[E_t]&\leq M\left(1+\frac{\log\left(1/\alpha\right)}{\emph{kl}(\mu,\mu-\varepsilon)}\right)\label{eq:expected_handicap_bound}
\end{align}
where $\emph{kl}(\mu,\mu-\varepsilon)$ is the Kullback-Leibler divergence between Bernoulli distributions \eqref{eq:kl-div}.
\end{theorem}
\begin{proof}
The proof follows from Lemma \ref{lemma:unsafe_sequences_escape} and Lemma \ref{lemma:sprt_detection_time}.
\end{proof}
In the same spirit as for the flawless setting, we now state that under a uniform strategy Algorithm \ref{alg:relaxed_inspector} finds all unsafe machines in expected finite time.
\begin{theorem}\label{thm:mab-relaxed-finite-stop-time}
Consider a $\lambda$-soft strategy $\psi$. Let $T$ be the time it takes for Algorithm \ref{alg:relaxed_inspector} to detect all the unsafe machines under $\psi$. Then:
$$\mathbb{E}[T]\leq \frac{{M}(K-M+1)}{\lambda}\left(1+\frac{\log
(1/\alpha)}{\emph{kl}(\mu,\mu-\varepsilon)}\right)\,.$$
\end{theorem}
\begin{proof}
The proof can be found in the Appendix \ref{app:mab-relaxed-finite-time}.
\end{proof}

We end this section by \emph{uniting} the flawless and relaxed settings, arguing that Algorithm \ref{alg:flawless-inspector} can be seen as a particular case of the SPRT used in Algorithm \ref{alg:relaxed_inspector}.

\begin{proposition}
Given fixed $\mu\in(0,1)$ and $\varepsilon=\beta\mu$. When $\beta\rightarrow1^-$, Algorithm \ref{alg:relaxed_inspector} with any  $\alpha>0$ reduces to Algorithm \ref{alg:flawless-inspector}. 
\end{proposition}
\begin{proof}
The flawless setting only allows for perfectly safe machines, discarding any arm at first sign of damage $D_t=1$. We will show that this coincides with a Sequential Probability Ratio Test that compares $\mathcal{H}_0:\mu_a\leq 0$~vs.~ $\mathcal{H}_1:~\mu_a>\mu$. This will hold for any $\mu\in(0,1)$ and for all $\alpha>0$.

For a fixed $\mu$, consider the test defined in \eqref{eq:hyp-test} with $\varepsilon=\beta\mu$. As $\beta\rightarrow 1^-$, the null hypothesis  $\mathcal{H}_0$ becomes $\mu_a=0$. We show that the log-likelihood ratio $\Lambda_a(t)$ goes to infinity at first sign of damage, thus declaring $\mathcal{H}_1$ and terminating the SPRT. A sufficient statistic for computing the log-likelihood ratio $\Lambda_a(t)$ is counting the $k$ outcomes of $D_\tau=1$ in a total of $t$ pulls. Then $\Lambda_a(t)=k\log\frac{\mu}{\mu-\varepsilon}+(t-k)\log\frac{1-\mu}{1-\mu+\varepsilon}$. Writing $\varepsilon=\beta\mu$,  $\Lambda_a(T)=t\log\frac{1-\mu}{1-\mu(1-\beta)}+k\log\frac{1-(1-\beta)\mu}{(1-\beta)(1-\mu)}\underset{\beta\to 1^-}{\longrightarrow}\infty~ \forall k>0$.
Then the SPRT decides on the alternative $\mathcal{H}_1$ at first sign of damage, no matter how small $\alpha$ is.
\end{proof}
\section{Assured Reinforcement Learning}\label{sec:assured_rl}
\hl{This section builds on the insights given by the MABs to detect and discard unsafe policies in the context of RL. Once more, we focus  on two different settings. First the \emph{flawless setting}, which seeks policies whose probability of encountering damage at any time $(D_t=1)$ is zero. Then, the \emph{relaxed setting}, which allows for a limited number of unsafe events in a single trajectory.}

\subsection{Problem formulation and outline}
Consider a  Markov Decision Process $\mathcal{M}$ with finite state space $\mathcal{S}$,  finite action space $\mathcal{A}$,  a  reward set $\mathcal{R}$, and a damage indicator $D_t\in\{0,1\}$.  A transition kernel $p$ specifies the conditional transition probability $p(s',r,d \mid s,a) :=P(S_{t+1}=s',R_{t+1}=r,D_{t+1}=d \mid S_t=s,A_t=a)$, from state $s\in\mathcal{S}$ and under action $a\in\mathcal{A}$, to state $s'\in\mathcal{S}$ resulting in a reward $r\in\mathcal{R}$ and a damage indicator $d$. \hl{We define the \emph{assured reinforcement learning problem} as follows:}
\begin{subequations}
\begin{align}
    \hl{V^*(s):=}&\hl{\max_\pi \mathbb{E}_{\pi}\left[\sum_{t=0}^\infty\gamma^t R_{t+1} ~\big|~ S_0=s\right]}\label{eq:assured-rl1}\\
    \hl{\text{s.t.:}}&\quad \hl{\mathbb{P}_\pi\left(\sum_{t=0}^\infty D_{t+1}\leq M\mid S_0=s\right)=1}
    \label{eq:assured-rl2}
\end{align}
\label{eq:assured-rl}
\end{subequations}%
\hl{\noindent 
where $0<\gamma<1$ is a discount factor, $M\in\mathbb{N}$, and $\mathbb{E}_\pi$ and $\mathbb{P}_\pi$ denote the expectation and probability with respect to the distribution induced by $\pi$. We call this an assured RL problem because \eqref{eq:assured-rl2} is a constraint that must be satisfied w.p.1. That constraint reads as: ``under policy $\pi$, starting from $s$, no trajectory can accumulate more than $M$ units of damage''. We call this quantity $M$ an \emph{allowable budget}.}

\hl{Analyzing this novel constraint calls for some work, and we subdivide this in two scenarios. The case $M=0$ (the flawless setting) will be considered first, and we give great breadth to its analysis.  We show how to obtain feasible policies in this case by developing a barrier function (Section \ref{sec:decomposition}) that encodes for constraint satisfaction. This barrier can be learned independently of the reward process (Theorem \ref{thm:decomposition}), and we show that the optimal barrier function can be attained  in expected finite time under a Barrier-learner algorithm (sections \ref{ssec:3.2} and \ref{sec:sample-compl}). Since the optimal barrier characterizes the set of feasible policies (Remark \ref{rmk:B_opt}), all feasible policies are found in finite time (similarly as in the previous section). We then present an \emph{assured} version of the Q-learning algorithm that makes use of this barrier. Finally, in   Subsection \ref{sec:relaxed-rl} we  address the \emph{relaxed setting} ($M>0$), showing that it can be reduced to the flawless setting in a suitably augmented MDP.}


\subsection{Value function decomposition}\label{sec:decomposition}
\hl{As argued previously, we will focus firstly on the flawless setting ($M=0$ in \eqref{eq:assured-rl}), which amounts to:}
\begin{subequations}
\begin{align}
    \hl{V^*(s):=}&\hl{\max_\pi \mathbb{E}_{\pi}\left[\sum_{t=0}^\infty\gamma^t R_{t+1} ~\big|~ S_0=s\right]}\label{eq:flawless-assured-rl1}\\
    \hl{\text{s.t.:}}&\quad \hl{\mathbb{P}_\pi\left(\sum_{t=0}^\infty D_{t+1}\leq 0\mid S_0=s\right)=1}
    \label{eq:flawless-assured-rl2}
\end{align}
\end{subequations}
\hl{Notice that, since $D_t$ only takes values $0$ or $1$, \eqref{eq:flawless-assured-rl2} can be equivalently put in either of the following two ways:}
\begin{align}
    \eqref{eq:flawless-assured-rl2}
    &\hl{\iff\mathbb E_{\pi}\left[\sum_{t=0}^\infty D_{t+1}~\big|~S_0=s\right]\leq 0}\label{eq:cumulative_constrained_problem}\\
    &\hl{\iff\mathbb{P}_\pi\left(D_{t+1}=0~\big|~S_0=s\right)=1\quad\forall t. }\label{eq:a.s.constraint}
\end{align}
The representation \eqref{eq:cumulative_constrained_problem} is a cumulative constraint in expectation, which could be put in Lagrangian form (see e.g., \cite{2019.Paternainyfa}\cite{2020.Ding}) in order to solve a primal-dual problem. We take an alternative approach that will lead to finite time detection, and resort to \eqref{eq:a.s.constraint} instead. Our goal then is to solve:
\begin{subequations}
\begin{align}
    V^*(s)&:= \max_\pi \mathbb{E}_{\pi}\left[\sum_{t=0}^\infty\gamma^t R_{t+1} ~\big|~ S_0=s\right]\label{eq:CRL3}\\
    \text{s.t.:}&\quad \mathbb{P}_\pi\left(D_{t+1}=0~\big|~S_0=s\right)=1~~\forall t\label{eq:general_formulation_constrained}
\end{align}
\label{eq:general-form}
\end{subequations}
\hl{With \eqref{eq:general_formulation_constrained} as constraint, there is a natural way to extend the definition of \emph{exposure} for MDPs.
\begin{definition}[Exposure for MDPs]\label{def:exposure-mdps1}
\hl{Consider an algorithm that creates a sequence of state-action pairs $(s_i,a_i)_{i=1}^t$. The exposure at time $t$ is:}
$$
\hl{E_t\!:=\!\sum_{i=1}^t\!\mathds{1}\!\left\{\!\min_\pi\mathbb{P}_{\pi}\!\left(\bigcup_{\tau=0}^\infty\{D_{\tau+1}\!=\!1\!\mid\! S_0\!=\!s_i,A_0\!=\!a_i\}\right)\!>\!0\right\}}.
$$
\end{definition}
\hl{Intuitively, the exposure at any step $i$ is equal to one if conditioned on starting at $(s_i,a_i)$ the constraint  \eqref{eq:general_formulation_constrained} does not hold for any policy. As such, it indicates whether a state-action pair $(s_i,a_i)$ is \emph{unsafe}.}~\hl{As in the bandits case, we will show that this quantity can be bounded in expected value.}\\
\hl{Although the expression above may seem cumbersome, we will show that it can be equivalently put in terms of a hard-barrier function that encodes safety, which we develop next.}}

With the formulation of \eqref{eq:general-form}, let us  define the value function $V^\pi$ for a specific policy $\pi$, in which the constraints are embedded inside the expectation: 
\begin{align}
    V^\pi(s) &:= \mathbb{E}_{\pi}\left[\sum_{t=0}^\infty\big(\gamma^t R_{t+1}+\mathbb I\left [D_{t+1}\right]\big) ~\big|~ S_0=s\right]\label{eq:v_pi}
\end{align}
where the hard barrier index function   $\mathbb I\left [\cdot\right]$ takes the form: 
\begin{equation}\mathbb I\left [D_{t+1}\right]=\log\left(1-D_{t+1}\right)=\begin{cases}0&\text{if } D_{t+1}=0\\-\infty&\text{if } D_{t+1}=1\end{cases}\label{eq:barrier_indicator}\end{equation}
so that it is null when the transition is safe, and takes the value $-\infty$ in the event of an unsafe transition. {Being that  \eqref{eq:barrier_indicator} is  unbounded,  expectations are defined in the sense of the Lebesgue integral for functions  in the extended real line~\cite{1999.Folland}.}

The proposed value function definition \eqref{eq:v_pi} will prove useful in two senses: firstly, we will show that maximizing \eqref{eq:v_pi} is the same as \eqref{eq:CRL3}--\eqref{eq:general_formulation_constrained}. 
Secondly, the additional term in \eqref{eq:v_pi} will allow for a barrier-based decomposition of the value function, which will aid in the learning of constraints.
%
%
%
%

\begin{lemma}[Equivalence]
\label{lemma:equivalence}
Problem \eqref{eq:general-form} is equivalent to the maximization of \eqref{eq:v_pi}, that is
\begin{align}
    \max_\pi \mathbb{E}_{\pi}\left[\sum_{t=0}^\infty\big(\gamma^t R_{t+1}+\mathbb I\left[D_{t+1}\right]\big) ~\big|~ S_0=s\right]\label{eq:CRL2}
\end{align}
\begin{proof}
If a policy $\pi_0$ is unfeasible for Problem \eqref{eq:general-form}, then $\exists t : P\left(D_{t+1}=1\right)>0$. This non-zero probability renders the expected value in \eqref{eq:CRL2} to $-\infty$ for $\pi_0$.
Conversely, if a policy $\pi_1$ attains a finite objective for \eqref{eq:CRL2}, then it must necessarily hold that $D_{t+1}=0$ almost surely $\forall t$, and hence $\pi_1$ is feasible for \eqref{eq:general-form}. \hl{Therefore the feasible set of \eqref{eq:general-form} coincides with the set of policies that obtain a finite value for \eqref{eq:CRL2}.} Lastly, any policy in any of these sets must satisfy $\log\left(1-D_{t+1}\right)=0 ~\forall t$, almost surely, in which case the function being maximized is the same. Then the optimal sets of the two problems coincide. 
\end{proof}
\end{lemma}
While solving \eqref{eq:general-form} is of our utmost interest, we have just shown that, to this end, we can solve \eqref{eq:CRL2} instead. In what follows we will take this one step further, and show that \eqref{eq:v_pi} admits a \textit{barrier-based decomposition} and can be cast as the sum of two value functions: one that checks only whether the policy in consideration is feasible (which will be the main focus of this work) and one that optimizes the return, provided the policy is feasible. The main idea behind this decoupling being that the search for feasible policies will be, in practice, an easier task to undergo.\\ 
To this end we define an auxiliary \textit{hard-barrier} value function  $H^\pi$ that will relate to $V^\pi$:
\begin{align}
    H^\pi(s) &= \mathbb{E}_{\pi}\left[\sum_{t=0}^\infty \log\left(1-D_{t+1}\right) ~\big|~ S_0=s\right]\label{eq:f_pi}
\end{align}
We proceed similarly for the action-value function $Q^\pi$ and its barrier  counterpart $B^\pi$:
\begin{align}
    &Q^\pi(s,a)\! =\! \mathbb{E}_{\pi}\!\left[\sum_{t=0}^\infty\big(\gamma^t R_{t+1}+\mathbb I\left[D_{t+1}\right]\big) \big| S_0=s, A_0=a\right]\nonumber\\
    &B^\pi(s,a) \!=\! \mathbb{E}_{\pi}\!\left[\sum_{t=0}^\infty\log\left(1-D_{t+1}\right) \big| S_0=s, A_0=a\right]\label{eq:b_pi}
\end{align}
Our original goal is to find  policies that are optimal for \eqref{eq:CRL2} for each possible state. By contrast,  maximizing \eqref{eq:f_pi} aims to find \textit{safe} policies, in the sense that they achieve a finite value in \eqref{eq:CRL2}. The main idea underpinning our work is that we can jointly work on optimizing \eqref{eq:f_pi}, which reduces the search over the policy space, while at the same time maximizing the return present in \eqref{eq:CRL2}.
In the following Theorem we establish a fundamental separation principle between the value functions and their auxiliary counterparts.

\begin{theorem}[Separation principle]
\label{thm:decomposition}
 Assume rewards $R_{t+1}$ are bounded almost surely for all $t$ \hl{and the discount factor satisfies $\gamma<1$}. Then, for every policy $\pi$: 
 \begin{equation}
     V^\pi(s) = V^\pi(s) + H^\pi(s) \label{eq:v_decomposition}
 \end{equation}
 \begin{equation}
     Q^\pi(s,a) = Q^\pi(s,a) + B^\pi(s,a) \label{eq:q_decomposition}
 \end{equation}
\begin{proof}
We shall prove \eqref{eq:v_decomposition} only, since  the proof  for \eqref{eq:q_decomposition} is alike. The following identities hold, as explained below.
\begin{align}
    &V^\pi(s) =\mathbb{E}_{\pi}\left[\sum_{t=0}^\infty\big(\gamma^t R_{t+1}+\log\left(1-D_{t+1}\right)\big)  ~\big|~  S_0=s\right]\nonumber\\
    &= \mathbb{E}_{\pi}\left[\sum_{t=0}^\infty\big(\gamma^t R_{t+1}+\log\left(1-D_{t+1}\right)\big)  ~\big|~  S_0=s\right]\label{eq:proof1_2}\\
    &+ \mathbb{E}_{\pi}\left[\sum_{t=0}^\infty\log\left(1-D_{t+1}\right) \mid S_0=s\right]\label{eq:proof1_3}
\end{align}
To show that $V^\pi(s)$ can be separated  in \eqref{eq:proof1_2} and \eqref{eq:proof1_3}, first suppose  policy $\pi$ is feasible for Problem \eqref{eq:CRL2}, in the sense that it attains a finite-valued objective. 
This necessarily implies that $D_{t+1}=0~~\textit{a.s.}~~\forall t$, which makes the second term in \eqref{eq:proof1_2} vanish. Conversely, suppose that the policy in consideration is infeasible. This together with the fact that rewards are bounded almost surely  yields $V^\pi(s)=-\infty$, which is the same value attained by both \eqref{eq:proof1_2} and \eqref{eq:proof1_3}.
\end{proof}
\end{theorem}

The preceding result implies non-trivial consequences. If the learning agent can interact with the environment and have access to rewards $R_{t+1}$ and queries of whether a transition has been safe (i.e. $D_{t+1}=0$), then it can separately learn both $Q^\pi(s,a)$ and $B^\pi(s,a)$. 
This is  discussed in the next remark.

\begin{remark}[Properties of the optimal barrier function $B^*$]
\label{rmk:B_opt}
\begin{align}
&B^*(s,a)=\max_\pi B^\pi(s,a)\nonumber\\ 
&= \max_\pi\mathbb{E}_{\pi}\left[\sum_{t=0}^\infty\log\left(1-D_{t+1}\right) ~\big|~ S_0=s, A_0=a\right] \label{eq:b_pi_opt}
\end{align}
\hl{By definition, the entries of $B^*$ will either be $0$ or $-\infty$. Having $B^*(s,a)=-\infty$ means that starting from $(s,a)$, no policy is safe, in the sense that there is an (albeit small) non-zero probability of encountering an unsafe event $D_{t+1}=1$, regardless of the actions taken later. Conversely, if $B^*(s,a)=0$ then, upon starting from $(s,a)$, there exists (at least one) policy that guarantees that no damage will be seen in the future. As such, the optimal barrier function $B^*$ completely characterizes the set of feasible, stationary policies $\Pi_{\texttt{safe}}$:
\begin{equation}
\Pi_{\texttt{safe}}=\left\{\pi: \pi(a|s)=0\text{~whenever~}B^*(s,a)=-\infty\right\}
\label{eq:set-feasible-policies}
\end{equation}
In this sense, prior learning of $B^*$ helps constrain the search of other known algorithms to only feasible policies.} 

\end{remark}
\hl{Since the optimal barrier function characterizes the safety of any state-action pair, it can be used to express exposure for MDPs more succintly (c.f. Definition \ref{def:exposure-mdps1}):}
\hl{\begin{remark}[Exposure in MDPs: alternate representation]\label{def:exposure-mdps2}
The exposure at time $t\geq 1$ is:
$$
\hl{E_t = \sum_{\tau=1}^t\mathds{1}\left(B^*(S_\tau,A_\tau)=-\infty\right)}.
$$
\end{remark}
}



 \subsection{Barrier learning}\label{ssec:3.2}

We now focus on the feasibility problem of learning $B^*$ from data.  To that end, we state the following optimality condition in the standard form of the Bellman's equations.

\begin{theorem}[Bellman equation for $B^*(s,a)$]
\label{thm:Bellman}
 The optimal barrier function satisfies  
 \begin{align}
      B^*(s,a)=\mathbb E\bigg[&\mathbb I\left[D_{t+1}\right]+ \max_{a'\in \mathcal A} B^*(S_{t+1},a') \big| S_t\!=\!s,A_t\!=\!a\bigg],\label{eq:bellman_for_barrier}
 \end{align}
where the expectation is taken with respect to the damage and next-state transition given by the MDP.
\end{theorem}
\begin{proof}
 It follows from Proposition 4.1.1 in  \cite[pp. 217]{2012.Bertsekas} with the value function being \emph{minimized}, and assuming possibly unbounded, non-negative costs   $C(S_t,A_t)=-\mathbb{I}[D_{t+1}]$.
\end{proof}
Besides providing a certificate for optimality, the Bellman Equations for $B^*(s,a)$ hint towards a stochastic iterative algorithm to optimize \eqref{eq:b_pi_opt} from data, the same way the Q-learning  algorithm is derived from the standard unconstrained value function \cite{2012.Bertsekas}. We will elaborate on this stochastic algorithm next.
%
%
%
As introduced in \eqref{eq:b_pi_opt}, our goal is to learn a \emph{safe} policy $\pi$ with an associated optimal Barrier function $B^*(s,a)$ that encodes the trajectories that satisfy the constraints at all times w.p.1. 
For this purpose, we first devise the following iterative algorithm that attempts to reach a fixed point satisfying \eqref{eq:bellman_for_barrier}
\begin{align}
B_{k+1}(s,a)&\!=\!\mathbb E\left[\mathbb I\left[D_{t+1}\right]+ \max_{a'\in \mathcal A} B_k(S_{t+1},a') \big| S_t\!=\!s,A_t\!=\!a \right].\label{eq:ensemble_iteration_bk}
\end{align}

Next, we appeal to the stochastic approximation machinery \cite{robbins_monro} to drop the unknown expectations yielding a data driven version of \eqref{eq:ensemble_iteration_bk}. The resulting stochastic update is given next.
%
%
\begin{algorithm}[htbp]
\SetAlgoLined
\KwIn{$B$-function and $(s_t,a_t,s_{t+1},d_{t+1})$ tuple}
\KwOut{Barrier-function $B$}
\vspace{-.5cm}
$$B(s_t,a_t)\leftarrow B(s_t,a_t)
+\log(1-d_{t+1})+\max_{a'}B(s_{t+1},a')$$
\KwRet{$B$}
\caption{\texttt{barrier\_update}}
\label{alg:barrier}
\end{algorithm}

The update in Algorithm \ref{alg:barrier} incorporates the information carried in  $d_{t+1}$, which signals whether the constraint has been violated or not during the transition from time $t$ to $t+1$. Moreover, the update does not only  consider immediate violations, but also the future effect of the action $a_t$ that is summarized in the second term  $\max_{a'}B(s_{t+1},a')$. This \emph{bootstrapping} mechanism leverages on stationarity to  collect  damage information from all previous state transitions, and summarize it in $B(s_{t+1},a')$ which predicts long-run future effect of the state-action pairs at time $t+1$. Thus, by repeating the update in Algorithm \ref{alg:barrier} with new data coming from successive system interactions, an agent can synthesize the whole information about all past constraint violations in the barrier function $B(s,a)$ for unveiling the set of unsafe policies. 

We turn now to the details of this iterative algorithm and to its performance guarantees.

\subsection{Performance analysis of Barrier-Learner}\label{ssec:3.3}
\label{sec:sample-compl}First, we consider a simple barrier learner algorithm where one can query on any state-action pair and sample transitions $(s,a)\rightarrow(s',d)$ according to the MDP kernel. The barrier learner is shown in Algorithm \ref{alg:barrier_learner}. Our analysis shows that the expected queries/samples required until all unsafe state-action pairs are detected is finite.
\begin{algorithm}[!ht]
\KwData{Constrained Markov Decision Process $\mathcal{M}$}
\KwResult{Optimal action-value function $B^*$}
Initialize $B^{(0)}(s,a)=0, \forall (s,a)\in\mathcal{S}\times \mathcal{A}$\\
\For{$t=0,1,\cdots$}{
    Draw $(s_t,a_t)\sim \mathrm{Unif}(\{(s,a):{B}^{(t)}(s,a)\neq -\infty\})$\\
    Sample transition $(s_t,a_t, s'_t,d_t)$ according to $\prob\lp S_1=s'_t,D_1=d_t|S_0=s_t,A_0=a_t\rp$\\
    $B^{(t+1)} \leftarrow \texttt{barrier\_update(}B^{(t)},s_t,a_t,s'_t,d_t\texttt{)}$\\
}
\caption{Barrier Learner Algorithm}
\label{alg:barrier_learner}
\end{algorithm}

In an MDP, an $(s,a)$ pair is unsafe ($B^*(s,a)=-\infty$) if either it immediately causes damage, i.e. $\prob(D_1=d|S_0=s,A_0=a)>0$, or it transitions to an unsafe state, namely $\prob(S_1=s'|S_0=s,A_0=a)>0$ for some $s'$ with $B^*(s',a)=-\infty,\forall a\in\mathcal{A}$. As a result, when an $(s,a)$ is taken, one may observe the damage after several steps.
We let $L$ be the \emph{lag} of the MDP, which is the maximum steps one need to wait until observing the potential damage by taking an unsafe $(s,a)$ pair. The exact definition of $L$ is given in Appendix \ref{app:pf_sample_comp_barrier_learner}. Regarding Algorithm \ref{alg:barrier_learner}, we have the following result 
\begin{theorem}\label{thm:sample_comp_barrier_learner}
    Given an MDP, let $\rho>0$ be a lower bound on all non-zero transition probabilities: $\rho\leq P(s'|s,a)~\forall (s,a,s'): P(s'|s,a)>0$. Let $T$ be earliest time when Algorithm \ref{alg:barrier_learner} detects all unsafe state-action pairs, i.e. $T:=\min\{t:B^{(t)}=B^*\}$. Then we have:
    \begin{equation}
        \expc [T]\leq (L+1)\frac{|\mathcal{S}||\mathcal{A}|}{\rho}\log\lp{|\mathcal{S}||\mathcal{A}|}+1\rp \,,\label{eq:expected_bound_barrier_learner}
    \end{equation}
    where $L$ is the lag of the MDP.
\end{theorem}
\begin{proof}[Proof Sketch]
    Theorem \ref{thm:sample_comp_barrier_learner} is proved in three steps. We refer the readers to Appendix \ref{app:pf_sample_comp_barrier_learner} for the complete proof.
    
    First, we reformulate the algorithm so that the sampling process of $(s_t,a_t)$ is independent of the current progress of the $B^{(t)}$-function: at iteration $t$, one samples an $(s_t,a_t)$ pair uniformly from the entire $\mathcal{S}\times \mathcal{A}$ set, then the algorithm chooses to either accept or reject the sample depending on the value of $B^{(t)}(s_t,a_t)$. This way, we turn to study the number of acceptance by the reformulated algorithm before it detects all unsafe state-action pairs.
    Secondly, we provide a modified algorithm with more restrictions on accepting the sample compared to the reformulated algorithm. Such restrictions allow the modified algorithm to learn unsafe state-action pairs in multiple stages: at each stage, the algorithm only allows to detect a subset of unsafe state-action pairs, but the detection process in every stage can be viewed as a safe multi-arm bandits problem discussed in Section \ref{sec:multi_arm}.
    
    Lastly, we derive an upper bound on the expected number queries of the modified algorithm until successfully detecting all unsafe state-action pairs, based on Theorem \ref{thm:mab-assured-finite-time}, which is also an upper bound on $\expc [T]$.
\end{proof}
While this bound is admittedly loose (we use a lower-bound $\rho$ for the transition probabilities and a provably slower surrogate algorithm with more restrictions), it serves the purpose of upholding our main claim in the paper that  all unsafe policies can be detected in finite time. The resulting simplicity of this bound also lets us observe the fundamental factors adding to  the detection time $T$.  Specifically, with   larger  spaces $\mathcal S$ and $\mathcal A$   more exploration is needed, with a smaller $\rho$ or longer lag $L$ unsafe actions take longer to be revealed as damaging, all three factors adding to a longer detection time. A tighter bound is presented in Appendix \ref{app:pf_sample_comp_barrier_learner} when we prove Theorem \ref{thm:sample_comp_barrier_learner}. 

\hl{This theorem has many implications. First, recalling that the optimal barrier $B^*$ fully characterizes the set of feasible policies (Remark \ref{rmk:B_opt}), we obtain \emph{all} the feasible policies in finite time. Next, as a simple corollary, we have that the \emph{exposure} (which counts the number of unsafe interactions with the environment) is also finite in expectation.}
\begin{corollary}[Expected exposure in MDPs is finite]
\hl{The expected exposure (see Remark \ref{def:exposure-mdps2}) under Algorithm \ref{alg:barrier_learner} satisfies:}
$$
\hl{\mathbb{E}[E_t]\leq (L+1)\frac{|\mathcal{S}||\mathcal{A}|}{\rho}\log\lp{|\mathcal{S}||\mathcal{A}|}+1\rp}$$
\begin{proof}
    \hl{Note that the exposure at any time $t$ is at most $t$: $E_t\leq t$ (which amounts to always sampling an unsafe $(s,a)$ pair). In particular, at the termination time $T$ of the algorithm we have $E_T\leq T$, and therefore $\mathbb{E}\left[E_T\right]\leq \mathbb{E}[T]$. The result follows from \eqref{eq:expected_bound_barrier_learner}.}
\end{proof}
\end{corollary}
\hl{We now re-use Theorem \ref{thm:sample_comp_barrier_learner} to derive a sample-complexity bound, and next discuss that the dimensionality dependence is far better than what appears in current works. }
\hl{\begin{corollary}[Sample-complexity bound]
For any $0<\delta<1$, with probability at least $1-\delta$, Algorithm \ref{alg:barrier_learner} learns the optimal barrier function $B^*$ after at most $T_\delta$ steps, with
$$
T_\delta=\left(1+\log\frac{1}{\delta}\right)(L+1)\frac{|\mathcal{S}||\mathcal{A}|}{\rho}\log\lp{|\mathcal{S}||\mathcal{A}|}+1\rp.
$$
\begin{proof}
    The stopping time $T$ defined in Theorem \ref{thm:sample_comp_barrier_learner} is a sum of geometric random variables (see Appendix \ref{app:pf_sample_comp_barrier_learner}). We use the fact that it is bounded in expectation by \eqref{eq:expected_bound_barrier_learner} together with a tail-bound for sums of geometric random variables \cite[Corollary 2.4]{janson2018tail} with mean $1+\log 1/\delta$.
\end{proof}
\end{corollary}}
\hl{\begin{remark}
    The preceding corollary shows that the optimal barrier function---and hence the set of feasible policies---can be learned with $\mathcal{O}\left(\log\frac{1}{\delta}(L+1)\frac{|\mathcal{S}||\mathcal{A}|}{\rho}\log|\mathcal{S}||\mathcal{A}|\right)$ samples. This is much more efficient than learning an $\epsilon$-optimal policy\cite{li2020breaking}, which requires  $\mathcal{O}\left(\frac{|\mathcal{S}||\mathcal{A}|}{(1-\gamma)^3\epsilon^2}\log\frac{|\mathcal{S}||\mathcal{A}|}{(1-\gamma)\epsilon\delta}\right)$ samples. Learning  an optimal policy requires good \emph{estimation} of the optimal value at each state, which brings the terms $\frac{1}{\epsilon}$ and $\frac{1}{1-\gamma}$ into play. Our framework, in contrast, cares only for \emph{detection} of unsafe state-action pairs, which can be done much faster.
\end{remark}
}
\hl{Once the optimal barrier function is known, one can optimize a policy over a smaller region, which corresponds to the reduced set of safe states $\mathcal{S}_{\texttt{safe}}$ and safe actions $\mathcal{A}_{\texttt{safe}}$. If the cardinality of these sets is much smaller than $\mathcal{S}$ and $\mathcal{A}$, the optimization process will be indeed faster. Along this line, we now illustrate how to learn both the safe region and an optimal policy by combining the barrier-learner with Q-learning.}





\subsection{Learning safely: Q-learning with a barrier}
Now that we have introduced this data-driven strategy for securing the environment, let us turn back our attention to the implications of the separation principle in Theorem \ref{thm:decomposition}. The separation principle \eqref{eq:q_decomposition} can be used to embed the safety information provided by $B(s,a)$ in the Q-function $Q(s,a)$.  Hence, we identify that the condition  $Q(s,a)=-\infty$ is equivalent to $B(s,a)=-\infty$ 
and this propagates information about safety from successor states. 
For those which do satisfy the constraints we have $B(s,a)=0$, and then $Q(s,a)$ will carry the information of the observed rewards. The following update is complementary to Algorithm \ref{alg:barrier_learner}, and amounts to the standard Q-learning algorithm for maximizing rewards at safe pairs of states and actions.

\begin{algorithm}[htbp]
\SetAlgoLined
\KwData{Step size $\eta$, discount factor $\gamma$}
\KwIn{ Functions $Q$, $B$,  and  $(s_t,a_t,s_{t+1},r_{t+1})$}
\KwOut{Q-function}
\vspace{-.5cm}
\begin{align*}
\hspace{-0.5cm}\hspace{-0.5cm}
Q(s_t,a_t)&\!\leftarrow\! (1\!-\!\eta)Q(s_t,a_t)\!+\!\eta\!\left(r_{t+1}\!+\!\gamma\max_{a'}Q(s_{t+1},a')\!\right)\\
Q(s_t,a_t)&\!\leftarrow\!B(s_t,a_t)+Q(s_t,a_t)
\end{align*}
\KwRet{$Q$}
\caption{\texttt{q\_update}}
\label{alg:reward_update}
\end{algorithm}

Next we specify how to use Algorithm \ref{alg:barrier} to ensure safety, and demonstrate the sample complexity of the learning process.  After that, we combine algorithms \ref{alg:barrier} and \ref{alg:reward_update} with the goal of  safely maximizing rewards. We will borrow the well-known convergence results of Q-learning \cite{1994.Tsitsiklis} together with our separation principle in Theorem \ref{thm:decomposition}  to provide performance guarantees for our novel Assured Q-learning algorithm.

The Barrier Learning Algorithm \ref{alg:barrier_learner} stands alone as a data-driven method to assess safety feasibility. However, our separation principle allow us to combine it with standard existing reward optimization algorithms in order to add safety.
For instance, by wrapping    Algorithm \ref{alg:barrier_learner} around the acclaimed Q-learning algorithm we obtain a Generative Assured Q-learning  method to maximize the rewards over the set of safe policies, as is shown in Algorithm \ref{alg:assured_qlearning_onebyone}.

\begin{algorithm}[!ht]
\KwData{Constrained Markov Decision Process $\mathcal{M}$}
\KwResult{Optimal action-value functions $B^*$ and $Q^*$}
Initialize $B^{(0)}(s,a)=0, \ Q^{(0)}(s,a)=0 \forall (s,a)\in\mathcal{S}\times \mathcal{A}$\\
\For{$t=0,1,\cdots$}{
    Draw $(s_t,a_t)\sim \mathrm{Unif}(\{(s,a):\mathcal{B}(s,a)\neq -\infty\})$\\
    Sample transition $(s_t,a_t, s'_t,d_t)$ according to $\prob\lp S_1=s'_t,D_1=d_t|S_0=s_t,A_0=a_t\rp$\\
    $B^{(t+1)} \leftarrow \texttt{barrier\_update(}B^{(t)},s_t,a_t,s'_t,d_t\texttt{)}$\\
$Q^{(t+1)}\leftarrow \texttt{q\_update(}B^{(t)},Q^{(t)},s_t,a_t,s'_t,r_t\texttt{)}$\\
}
\caption{Generative Assured Q-Learning }
\label{alg:assured_qlearning_onebyone}
\end{algorithm}

As a corollary of Theorem \ref{thm:sample_comp_barrier_learner}, and borrowing the convergence results of Q-learning from \cite{1994.Tsitsiklis} we establish the following result.

\begin{corollary}\label{crl:assured-converges}
With  finite state space $\mathcal S$ and action space $\mathcal A$,   bounded rewards  $R_{t}\leq C$, and diminishing step-sizes satisfying $\sum_t\eta_t=\infty$ and $\sum_t\eta_t^2<\infty$, 
the iterates $Q^{(t)}$ of the Algorithm \ref{alg:assured_qlearning_onebyone} converge almost surely to the optimal Q-function $Q^\star$ satisfying:
$$\lim_{t\to\infty} Q^{(t)}(s,a)=Q^*(s,a)\ (w.p.1)~~ \forall (s,a)\in \mathcal S\times \mathcal A$$
with:
\begin{align*}
      Q^*(s,a)\hspace{-2pt}=\hspace{-2pt}\mathbb E\bigg[&\hspace{-2pt}\log\left(1\hspace{-2pt}-\hspace{-2pt}D_{t+1}\right)
      \hspace{-2pt}+\hspace{-2pt} \max_{a'\in \mathcal A} Q^*(s',a') \big|S_t=s,A_t=a\bigg].
\end{align*}
 
  Moreover, all unsafe state-action pairs corresponding to $Q^*(s,a)=-\infty$ are detected in expected finite time $T_{s,a}$ such that: $$Q^{(t)}(s,a)=-\infty,\ \forall\ t\geq T_{s,a},$$ and:
 \begin{equation}
        \expc \left[T_{s,a}\right]\leq \frac{|\mathcal{S}|^2|\mathcal{A}|}{\mu}\log\lp{|\mathcal{S}||\mathcal{A}|}+1\rp \,.\label{eq:finite_time_assuredq}
    \end{equation}
 \end{corollary}

\begin{proof}
In order to apply the convergence results of \cite{1994.Tsitsiklis} we need iterates $Q^{(t)}$ to be finite for all $t$. But this is guaranteed by Theorem \ref{thm:sample_comp_barrier_learner} for all safe pairs $(s,a)$ such that $B^*(s,a)=0$. Indeed,    or all values such that $B^*(s,a)=0$, we have $B^{(t)}=0$ for all $t$ and the updates of the function $Q^{(t)}$ in Algorithm \ref{alg:assured_qlearning_onebyone} amount to the asynchronous Q-learning updates with finite rewards and diminishing step-size required by \cite{1994.Tsitsiklis}.  

For the pairs such that $B^*(s,a)=-\infty$, Theorem \ref{thm:decomposition} implies  $Q^*(s,a)=-\infty$  and according to Theorem \ref{thm:sample_comp_barrier_learner} there must exist a time instant $T_{s,a}$ satisfying \eqref{eq:finite_time_assuredq} such that $B^{(t)}=-\infty ~~\forall t\geq T_{s,a}$. Since by construction the q\_update in Algorithms \ref{alg:assured_qlearning_onebyone} and \ref{alg:barrier} implies that $Q^{(t)}=-\infty$ whenever $B^{(t)}=-\infty$, then $\lim_{t\to\infty}Q^{(t)}(s,a)=-\infty$ for all pairs $(s,a)$ such that     
$B^*(s,a)=Q^*(s,a)=-\infty$.
\end{proof}

The previous result applies to the specific case of diminishing step-sizes and immediate restarts after episodes of length one. However, Q-learning is widely applied with longer episodes and convergence is guaranteed provided that each state-action pair is visited infinitely often.
While one step episodes in Algorithm \ref{alg:assured_qlearning_onebyone} were used in order to simplify the proof of Corollary \ref{crl:assured-converges}, the numerical experiments of the next section will demonstrate that these safe convergence results  carry out to an episodic version  of Assured Q-learning with $\epsilon$-greedy exploratory  modes.

\subsection{Relaxed setting}\label{sec:relaxed-rl}
\hl{
We finish this section by analyzing the relaxed setting, which corresponds to $M>0$ in \eqref{eq:assured-rl}:
\begin{subequations}
\begin{align}
    \max_\pi &~\mathbb{E}_{\pi}\left[\sum_{t=0}^\infty\gamma^t R_{t+1} ~\big|~ S_0=s\right]\label{eq:budget-rl-goal}\\
    \text{s.t.:}&~ \mathbb{P}_\pi\left(\sum_{t=0}^\infty D_{t+1}\leq M\mid S_0=s\right)=1;\quad M>0\label{eq:budget-rl-constraint}.
\end{align}
\label{eq:budget-rl}
\end{subequations}
}
\hl{This setting differs from the one analyzed so far, in the sense that we allow \emph{at most} $M$ units of damage in any given trajectory.
To solve \eqref{eq:budget-rl} we will resort to state-augmentation, tracking how much damage an agent has received so far. We define the variable to augment into the state next.
}
\hl{
\begin{definition}[Safety budget]
The safety budget $K_t$ at time $t$ is defined and evolves as
\begin{subequations}
\begin{align}
    K_0:&= M\\
    K_{t+1}&= K_t-D_t \label{eq:budget-transition}
\end{align}
\end{subequations}
\label{eq:budget-def}
\end{definition}
Here, $K_t$ specifies the remaining budget at time $t$, that is, how many more units of damage an agent can sustain and still satisfy \eqref{eq:budget-rl-constraint}. We consider a state-augmented MDP $\tilde{\mathcal{M}}$ with state $\tilde{S}_t=(S_t,K_t)$ (notice with \eqref{eq:budget-transition} transitions are still Markovian), and observe that \eqref{eq:budget-rl} can be equivalently formulated as:}
\hl{
\begin{subequations}
\begin{align}
    &\max_{\tilde\pi} \mathbb{E}_{\tilde{\pi}}\left[\sum_{t=0}^\infty\gamma^t R_{t+1} ~\big|~ S_0=s, K_0=M\right]\\
    &\text{s.t.:}~ \mathbb{P}_{\tilde{\pi}}\left(K_{t+1}\geq 0\mid S_0=s,K_0=M\right)=1~~\forall t\geq 0 \label{eq:budget-rl-equiv-constraint}
\end{align}
\label{eq:budget-rl-equiv}
\end{subequations}
where the notation $\tilde\pi$ stresses that this is a policy on the augmented MDP $\tilde{\mathcal{M}}$. By defining a new binary damage signal $\tilde{D}_{t+1}$ on $\tilde{\mathcal{M}}$ we can rewrite the constraint \eqref{eq:budget-rl-equiv-constraint} as follows:}
\hl{\begin{equation}
    K_{t+1}\geq 0\iff \tilde{D}_{t+1}:=\mathds{1}\left\{K_{t+1}<0\right\}=0
\end{equation}}
\hl{Now, outstandingly, the previous state-augmented problem can be equivalently put in the following form, which \emph{fits} the flawless setting of the beginning of the section:}
\hl{\begin{subequations}
\begin{align}
    &\max_{\tilde\pi} \mathbb{E}_{\tilde{\pi}}\left[\sum_{t=0}^\infty\gamma^t R_{t+1} ~\big|~ S_0=s, K_0=M\right]\\
    &\text{s.t.:}~ \mathbb{P}_{\tilde{\pi}}\left(\tilde{D}_{t+1}= 0\mid S_0=s, K_0=M\right)=1~\forall t\geq 0 \label{eq:budget-rl-equiv-constraint2}
\end{align}
\label{eq:budget-rl-equiv2}
\end{subequations}}
\hl{We focus then on \eqref{eq:budget-rl-equiv2}, and state its main properties in the following theorem.}
\hl{
\begin{theorem}[Stationarity and equivalence]\label{thm:stat-equiv}
~
\begin{enumerate}
    \item If \eqref{eq:budget-rl-equiv2} is feasible, there exists an optimal policy that is stationary: that is $\tilde{\pi}^*(\cdot|s,k)$.
    \item In that case, $\tilde{\pi}^*$ is also optimal for \eqref{eq:budget-rl}.
\end{enumerate}
\begin{proof}
To prove (1), swap \eqref{eq:budget-rl-equiv-constraint2} for the equivalent constraint $\mathbb{E}_{\tilde{\pi}}\left[\sum_{t=0}^\infty \gamma^t \tilde{D}_{t+1} | S_0=s, K_0=M\right]\leq 0$. This fits the standard formulation for CMDPs, for which the set of stationary policies is \emph{complete} \cite{1998.Altman}. Then, if the problem is feasible, there exists at least one stationary optimal policy. To prove (2), we refer the reader to Lemma 5 and Lemma 6 in \cite{l4dc}, where the equivalence between \eqref{eq:budget-rl} and \eqref{eq:budget-rl-equiv2} is shown.
\end{proof}
\end{theorem}
}
\hl{Due to Theorem \ref{thm:stat-equiv} and the fact that \eqref{eq:budget-rl-equiv2} fits the formulation of the beginning of the section, we can similarly define an extended action-value function $\tilde{Q}^{\tilde\pi}(s,k,a)\! :=\! \mathbb{E}_{\pi}\!\left[\sum_{t=0}^\infty\big(\gamma^t R_{t+1}+\mathbb I\left[\tilde{D}_{t+1}\right]\big) \big| S_0=s, K_0=k,A_0=a\right]$, the hard-barrier $\tilde{B}^{\tilde{\pi}}(s,k,a)$,
and the separation principle still holds. Then, under this setting learning feasibility would be akin to learning the optimal barrier $\tilde{B}^*(s,k,a)\quad\forall (s,k,a)\in\mathcal{S}\times[M]\times\mathcal{A}$ where $[M]=\{0,1,\ldots,M\}$. This can still be done in expected finite time by applying the results of Theorem \ref{thm:sample_comp_barrier_learner}, along with finite exposure, defined in the extended MDP. There is a seemingly big price to pay in this case: the dimensionality increase of needing to learn $\tilde{B}^*$ for each possible budget $k$. This however, can be circumvented, and we refer the reader to \cite{l4dc} for details.}

\section{Numerical Experiments}\label{sec:numericals}
We now proceed to some numerical experiments that back up the results presented throughout the paper. We first focus on the multi-armed bandit setup and on the problem of detecting unsafe machines under a uniform strategy. Later we tackle learning the optimal barrier in a continuous control problem, and how learning this barrier---i.e. feasibility---first can make learning a task easier later.

\subsection{Multi-armed bandits}\label{sec:experiment-bandits}

We illustrate the performance of the Relaxed Inspector (Algorithm \ref{alg:relaxed_inspector}) on a setup of $K=1000$ arms, with a safety requirement of $\mu=\frac{1}{10}$. Each arm's true safety parameter is uniformly sampled between $0$ and $\frac{1}{5}$. This means that if $\epsilon\approx 0$, around half of the machines are safe and the other half unsafe. We run the Relaxed Inspector on this setting under a uniform strategy, for varying levels of both the failure tolerance $\alpha$ and the slack $\epsilon$. After all unsafe machines have been detected ---which happens and is certified to be done in finite time in virtue of Theorem \ref{thm:mab-relaxed-finite-time}--- we consider $C_{\varepsilon,\infty}$, the final conservation ratio and the normalized final exposure $\frac{1}{K}E_\infty$ .  Figure \ref{fig:bandits} shows both metrics for varying $\alpha$ and $\epsilon$. The curves in these figures show the average value obtained after $16$ independent runs. Shaded intervals correspond to $\pm \sigma/\sqrt{16}$, with $\sigma$ being the sample deviation.

These figures certify the intrinsic trade-off in our methodology: if a large value of $\alpha$ is used, one can obtain low exposure $E_t$ (less pulls of unsafe machines), but at the cost of discarding more safe machines (smaller conservation ratio $C_{\varepsilon,t}$). For further examples we refer the reader to \cite{castellano2021learning}.



\begin{figure}[ht]
    \centering
    \begin{subfigure}[t]{.49\columnwidth}
         \centering
    \includegraphics[width=\linewidth]{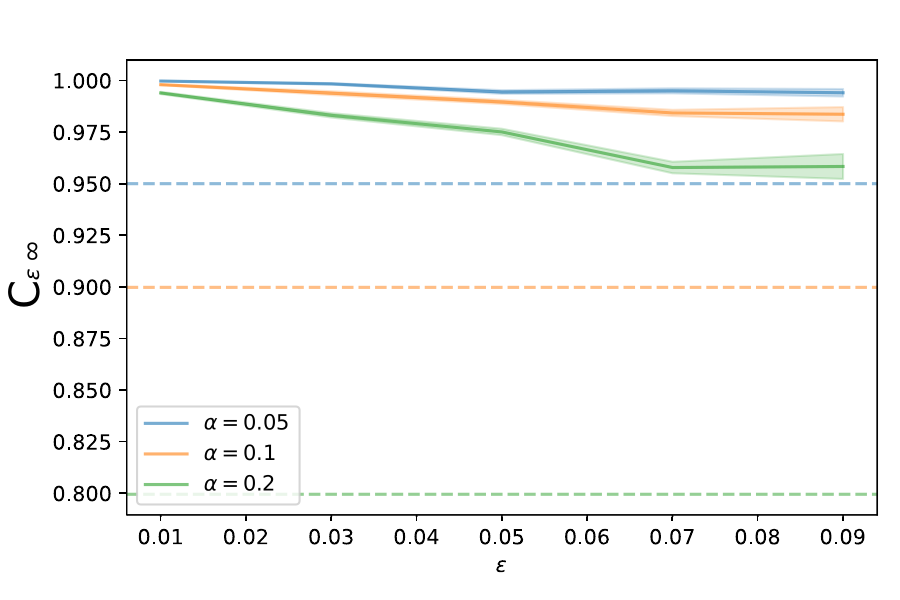}
    \caption{}
    \label{fig:bandits-rho}
     \end{subfigure}
     \begin{subfigure}[t]{.49\columnwidth}
     \includegraphics[width=\linewidth]{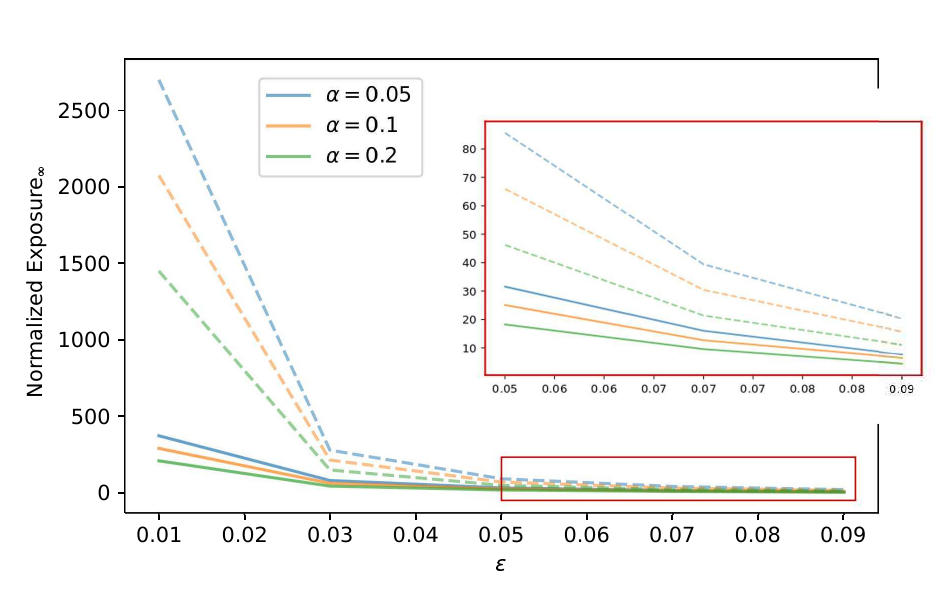}
    \caption{ }
    \label{fig:bandits-expo}
     \end{subfigure}
     \caption{Top: Final safety ratio $C_{\varepsilon,\infty}$ for $\mu=0.1$ as a function of $\epsilon$, for varying $\alpha$. More safe machines are kept when using small values of $\epsilon$ and $\alpha$, but this in turn implies longer training time. Bottom: Normalized final exposure for $\mu=0.1$ as a function of $\epsilon$, for varying tolerance level $\alpha$. Larger values of $\alpha$ and $\epsilon$ achieve lower handicap (which implies faster detection).
     }
     \label{fig:bandits}
\end{figure}

\subsection{MDPs}\label{sec:experiment-mdps}

\hl{We consider a robot navigating in 2d-space with constant velocity module $v=0.5m/s$. The robot can change its angle via first-order tracking dynamics, and the system evolves according to the following equations:
\begin{equation}
\left\{\begin{array}{ll}
\dot x = v.\cos\theta\\
\dot y = v.\sin\theta\\
\dot\theta = -(a-\theta)
\end{array}\right.\label{eq:sys_dynamics}
\end{equation}
where $x$ and $y$ are the   horizontal positions, $\theta$ is the angle with respect to the horizontal,
$s=[x, y, \theta]$ is the state of the system and the action $a$ is an  angle setpoint.\\
The $(x,y)$-space is a $5\times 5$ square with an obstacle in the middle, as depicted in the top-left of Fig. \ref{fig:env-and-barrier}. Bumping into either the obstacle or any of the outer walls is \emph{unsafe}, and results in damage $D_{t+1}=1$. 
At each time step, the agent decides an action $a\in[\theta_0-\pi, \theta_0+\pi]$ where $\theta_0$ is the angle at the beginning of the time interval. This action is input to the system \eqref{eq:sys_dynamics} for $\Delta T = 0.5s$. We uniformly discretize each component of the state-action space into $N_x=N_y=21, N_\theta=N_a=8$ values respectively.
}


\subsubsection{Learning the barrier}
\hl{We learn the barrier running episodes in which the initial state is  picked uniformly at random, and the episode  finishes either after $100$ timesteps or if the agent receives damage (bumps into the object or goes out of bounds). At each step, the agent takes a random action over the presumed safe ones: those for which $B(s,a)=0$.
Figure \ref{fig:env-and-barrier} shows the learned barrier after $2\times 10^6$ episodes.}

\begin{figure}[t]
    \centering
    \includegraphics[width=.8\linewidth]{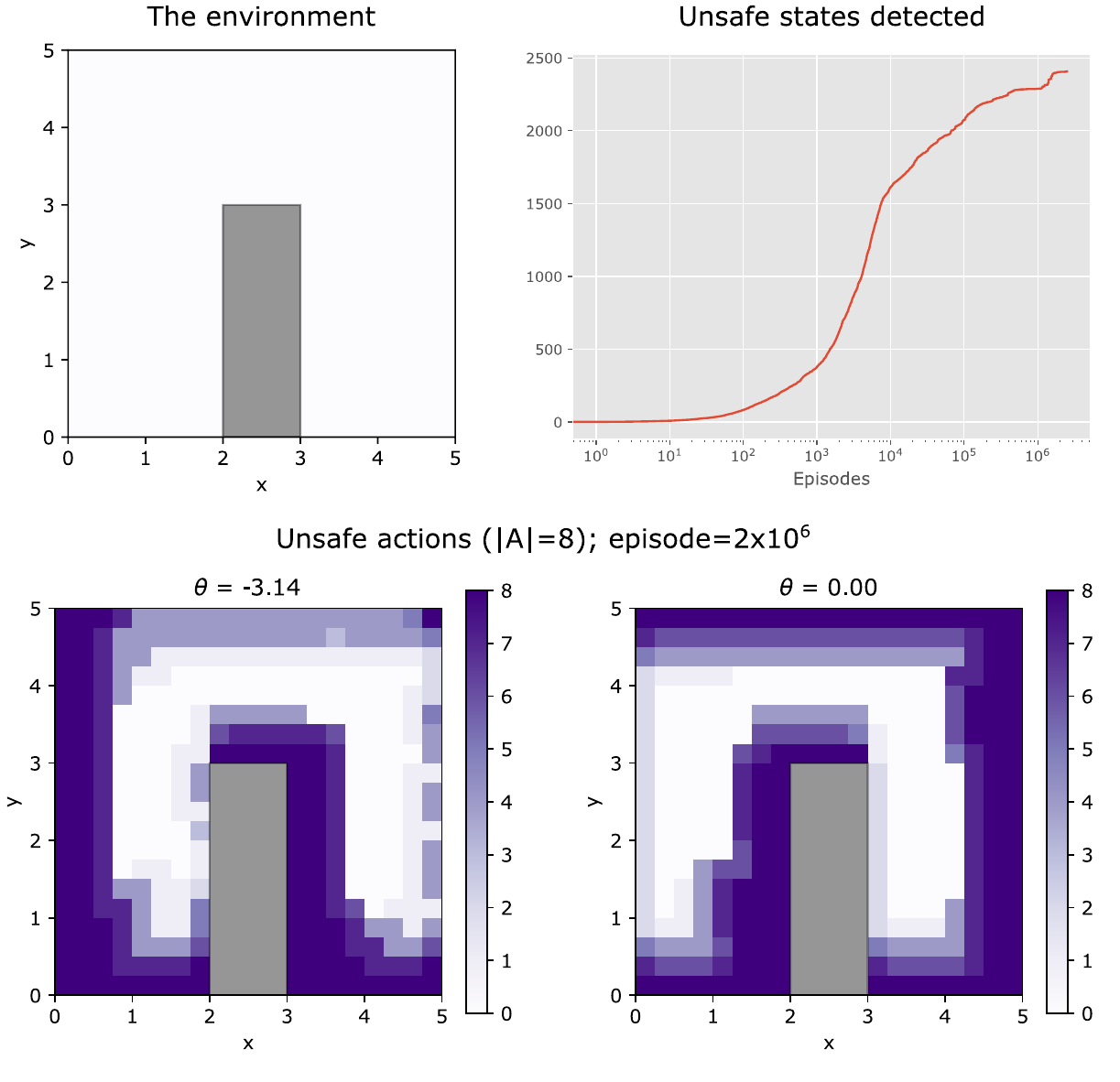}
    \caption{{Top-left: the environment. An agent navigates through it, and gets damage if it collides onto the gray obstacle or the outer walls. Top-right: learning curve for the unsafe states detected by the barrier learner (those for which $B(s,a)=-\infty~\forall a$). Bottom: visualization of the learned barrier after two million episodes. The left figure shows the barrier when the agent's angle is $\theta=-\pi$ (heading left); the right figure shows the barrier for an agent facing right. }}
    \label{fig:env-and-barrier}
\end{figure}


\subsubsection{Knowing the barrier accelerates learning}
\hl{In this task-oriented version of the previous environment, the goal is to reach a region of the space,  a circle of radius $0.5$ centered at coordinates $(4.5, 0.5)$. Each episode starts in position $(0.5,0.5)$ with heading $\theta=0$. At each time step the reward $R_t$ is minus the distance to the center of the goal. Bumping into the wall ends the episode with additional $-100$ reward. Reaching the goal ends the episode with additional $+100$ reward}

\hl{We compare a standard Q-learning agent versus its \emph{assured} counterpart: an agent that has previously learned the barrier of Fig. \ref{fig:env-and-barrier}, and that only takes (presumably) safe actions. At each step, the standard agent follows an $\epsilon$-greedy policy, while the assured agent follows a \emph{safe} $\epsilon$-greedy policy, only taking actions over the set of presumed-to-be safe actions ($B(s,a)=0$). At each step $t$ we observe a tuple $(S,A,S',R,D)$ and update the $Q-$function as $Q_{t+1}(S,A)= (1-\alpha)Q_t(S,A)+\alpha\left(R+\gamma\max_{a'}Q_t(S',A')\right)$. Both agents use $\epsilon=0.1$, $\alpha=0.1$, $\gamma=0.99$.}
\hl{Fig. \ref{fig:trajectories} shows $100$ greedy trajectories obtained with each agent during different stages of training. Assured Q-learning (on the left) rapidly learns to reach the goal, always succeeding. On the other hand, some trajectories of standard Q-learning (on the right) fail against the wall. The assured agent is conservative, leaving a bigger gap with the obstacle.}

\hl{Some quantitative metrics for these examples are shown in Fig. \ref{fig:boxplots}, that show the total reward and steps taken to reach the goal over these trajectories. Each column shows statistics using $100$ greedy trajectories as sample points. Each box spans the first ($q_1$) and third ($q_3$) quartile, with the median shown as a solid line. Whiskers above and below each box correspond to a confidence interval depicting the inter-quartile length $1.5\times(q_3-q_1)$, with outliers as black circles. Green triangles represent the mean values. The value $N$ indicates how many trajectories (out of $100$) reach the goal safely.}

\hl{Key takeaways from this experiment are that: \emph{i)} the assured agent always reaches the goal safely, while the standard one sometimes fails; \emph{ii)} when both reach the goal, the assured agent is typically faster; \emph{iii)} the assured agent is more \emph{conservative}, staying further away from the obstacles.}

\begin{figure}[t]
\centering
    \includegraphics[width=1\columnwidth]{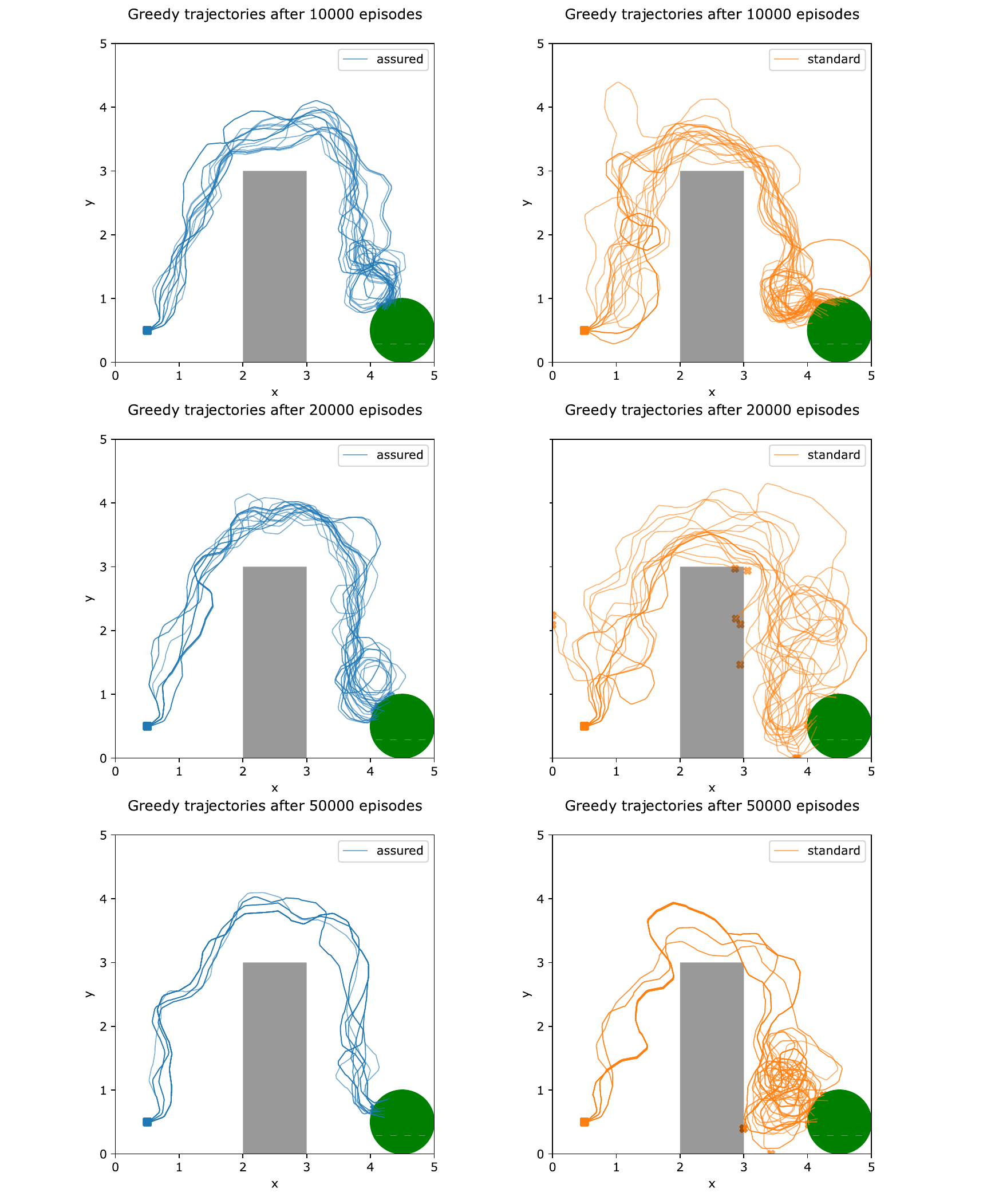}
     \caption{{Samples of greedy trajectories at different stages of training. The first column corresponds to the \emph{assured} agent, the second one to the standard Q-learning agent.The assured agent always reaches the goal while the standard agent still fails in some trajectories. Notice the conservative nature of the assured agent, who leaves a bigger margin between itself and the obstacle.}}
    \label{fig:trajectories}
\end{figure}



\begin{figure}[ht]
    \centering
    \begin{subfigure}[t]{.48\columnwidth}
         \centering
         \includegraphics[width=\textwidth]{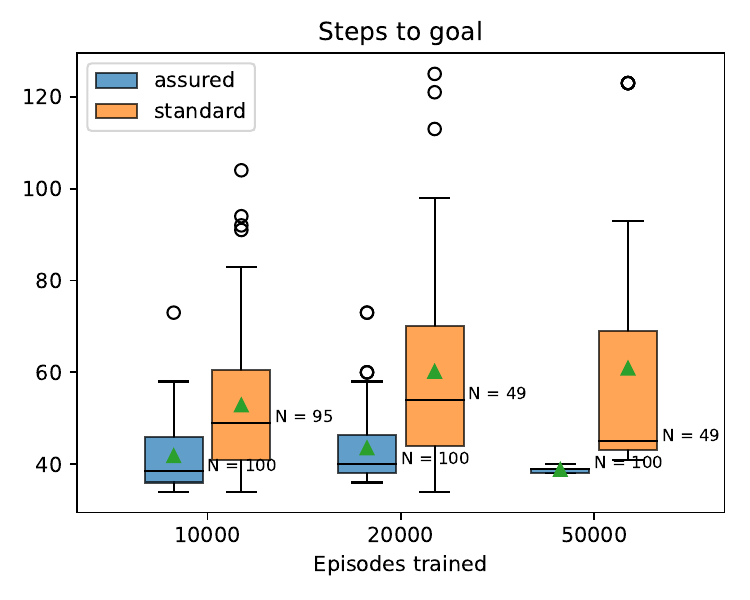}
     \end{subfigure}
     \begin{subfigure}[t]{.48\columnwidth}
         \centering
         \includegraphics[width=\textwidth]{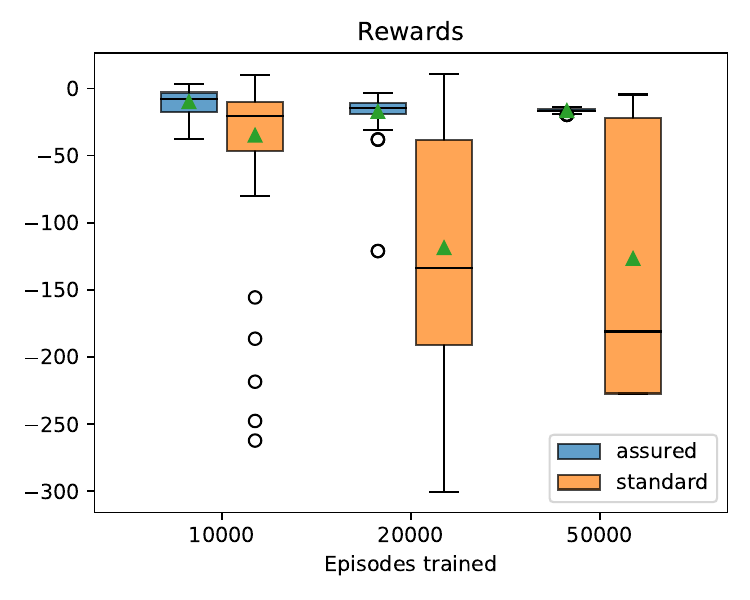}
     \end{subfigure}
     \caption{{Left: steps taken to reach the goal for both agents at different stages of training. The value $N$ shows how many trajectories (out of $100$) made it to the goal. The assured agent always reaches the goal ($N=100$), usually in fewer steps; the standard agent reaches the goal half of the time ($N=49$) after 20000 and 50000 episodes. Right: total reward along the trajectories. The assured agent's rewards are more tightly concentrated, and generally higher.
     }}
     \label{fig:boxplots}
\end{figure}
\section{Conclusions}
In this work we addressed the problem of learning to act safely in unknown environments. We made the case that learning safe policies is fundamentally different from learning optimal policies, and that it can be done separately and in a more efficient manner. By incorporating in the model a binary \emph{damage signal} that indicates constraint violations, we showed that identification of all unsafe actions (MABs) and state-action pairs (MDPs) is achieved in expected finite time with probability one guarantees. These results imply that the learner is not indefinitely exposed to damage, and could aid in the design of new algorithms that rapidly learn to act safely while jointly optimizing returns. \hl{Our experimental results for MDPs suggest that our algorithm obtains good performance in a continuous-state dynamical system, making it potentially useful for control applications.
}
\appendix
\subsection{Proof of Theorem \ref{thm:mab-assured-finite-time}}\label{app:mab-assured-unif}
{We will prove the following two inequalities:}
$${\mathbb{E}[T] \leq \frac{1}{\lambda \mu_{\text{low}}}\sum_{i=0}^{M-1}\frac{K-i}{M-i}\leq \frac{M + (K-M)\log(M+1)}{\lambda\mu_{\text{low}}}.}$$

\begin{proof}
    Each iteration of Algorithm \ref{alg:flawless-inspector} can be view as doing a Bernoulli trial with success rate being the probability of detecting an unsafe machine. {This probability evolves over time, and depends on the failure probability of each arm, and on the number of unsafe machines in the candidate set.}

{We can decompose the time $T$ as }
${T = T_1 + (T_2-T_1) + (T_3-T_2) + \ldots + T-T_{M-1},}
$
{where $T_i$ is the total time taken to detect the $i$-th machine. Then}
\begin{equation}\label{eq:time_diffs}
{\mathbb{E}[T] = \mathbb{E}[T_1] + \mathbb{E}[T_2-T_1] + \ldots + \mathbb{E}[T-T_{M-1}]}
\end{equation}
{We first bound $\mathbb{E}[T_1]$. When all $M$ malfunctioning machines are in play, $T_1$ is just the time taken to detect one of them. The probability of detecting the first machine is then}
\begin{align*}
{p_1}&{:=\mathbb{P}(\text{detect first machine})=\mathbb{P}(\text{get damage})}\\
&{=\sum_{a=1}^M\mathbb{P}(\text{get dmg}|\text{pull~}a)\mathbb{P}(\text{pull~}a)\geq \frac{\lambda}{K}\sum_{a=1}^M \mu_a}{\geq \frac{\lambda M\mu_{\text{low}}}{K}}
\end{align*}
{where in the first inequality we used the fact that the strategy $\psi$ is $\lambda$-soft, and in the second one the lower bound on $\mu_a$. Since $T_1\sim Geom(p_1)$, we thus have:}
$${\mathbb{E}[T_1] = \frac{1}{p_1}\leq \frac{1}{\lambda\mu_{\text{low}}}\cdot\frac{K}{M}}$$
{We now proceed to bound $\mathbb{E}[T_2-T_1]$. After detecting the first machine, now $p_2$ is the probability of observing damage when dealing with $M-1$ malfunctioning machines over a total of $K-1$ machines. Now $T_2-T_1\sim Geom(p_2)$ and we get:}
$$
{\mathbb{E}[T_2-T_1]\leq \frac{1}{\lambda\mu_\text{low}}\cdot\frac{K-1}{M-1}}
$$
{Proceeding similarly for the remaining stages, we end up bounding $\mathbb{E}[T]$ in \eqref{eq:time_diffs} as:}
$$
{\mathbb{E}[T] \leq \frac{1}{\lambda \mu_{\text{low}}}\sum_{i=0}^{M-1}\frac{K-i}{M-i}}
$$
{What remains to be shown is a further upper bound on this right hand side. To that end, we manipulate the sum:}
\begin{align*}
    &{\sum_{i=0}^{M-1}\frac{K-i}{M-i}}
    {= M  +(K-M)\sum_{i=0}^{M-1}\frac{1}{M-i}}\\
    &{= M+(K-M)\sum_{i=1}^M\frac{1}{i}}{\leq M + (K-M)\log(M+1)}
\end{align*}
{where on the last inequality we used the usual bound on the harmonic series $\sum_{i=1}^M\frac{1}{i}<\log(M+1)$.}
\end{proof}


\ifthenelse{\boolean{arxiv}}{
\subsection{Proof of Lemmas \ref{lemma:unsafe_sequences_escape}--\ref{lemma:sprt_detection_time}}\label{app:relaxed-lemmas}

\begin{figure}[t]
\begin{tikzpicture}[xscale=.9, yscale=.7]
    \input{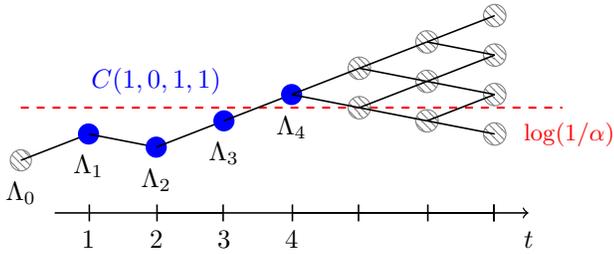}
\end{tikzpicture}
\caption{Visualization of the cylinder set $C(1,0,1,1)$. It contains all trajectories $\{x_\tau\}_{\tau=1,\ldots,\infty}$ that start with $1, 0, 1, 1.$ These 4 first points are shown in solid blue ($\Lambda_1$ through $\Lambda_4$). 
Each occurence of $d_t=1$ makes $\Lambda_i$ grow, while each $0$ makes it decrease.}
\label{fig:sprt_schematic}\end{figure}

In order to prove these lemmas we first review cylinder sets, which were first described by Wald \cite{wald1945}. Consider an infinite sequence $(d_\tau)_{\tau\geq 1}$ and define $C_\infty$ as the space of all such  sequences. The set $C(x_1,\ldots,x_t)$ is called a cylinder set of order $t$, and is defined as the subset of $C_\infty$ which collects sequences with  $d_1=x_1, \ldots, d_t=x_t$.
A cylinder set will be said to be of the \emph{unsafe type} if 

\begin{equation}\label{eq:define_cylinder1}
\Lambda_t=\log\frac{f_{\mu}\left(x_{1},\ldots,x_{t}\right)}{f_{\mu-\epsilon}\left(x_{1},\ldots,x_{t}\right)} \geq \log(1/\alpha)\,,
\end{equation}
and
\begin{equation}\label{eq:define_cylinder2}
\Lambda_\tau=\frac{f_{\mu}\left(x_{1},\ldots,x_{\tau}\right)}{f_{\mu-\epsilon}\left(x_{1}\ldots, x_{\tau}\right)}<\log(1/\alpha)\quad\forall \tau<t\;.
\end{equation}
The first condition guarantees that all sequences $\{d_\tau\}$ belonging to the unsafe cylinder $C(x_1,\ldots,x_t)$ will lead to the acceptance of $\mathcal{H}_1$. This will occur at time $t$, regardless of the future samples $d_\tau, \tau>t$. The second condition states that the threshold is not surpassed sooner than $t$. Conditions \eqref{eq:define_cylinder1}--\eqref{eq:define_cylinder2} imply that cylinders of different orders are disjoint sets, since $\Lambda_t$ exceeds the threshold for the first time at $t$ and no sooner than that, and this cannot be true for different values of $t$.

The union of all unsafe cylinder sets (of any order) defines the set of sequences that lead to deciding $H_1$.
 Let us  name this (disjoint) union as $Q_U$. 
%
%
%
%
%
%
%
 Let us also define  $Q_S$  as the complement of $Q_U$
\begin{equation}
    Q_S = Q_U^\complement\,.
\end{equation}
$Q_S$ is the set of all sequences for which $\Lambda_t<\log(1/\alpha)$ for all $t$. Since the sets are complementary, it holds that
\begin{equation}
P_{\mu_a}(Q_U+Q_S)=1 \quad\forall \mu_a\in[0, 1]\,,
\label{eq:pQ1S}
\end{equation}
with $P_{\mu_a}(Q)$ being the probability of any subset $Q$ of $C_\infty$, under the assumption that the sequence is generated by i.i.d. Bernoulli random variables with parameter $\mu_a$.

{We now turn to the proof of Lemmas  \ref{lemma:unsafe_sequences_escape}--\ref{lemma:sprt_detection_time} by restating them in terms of the cylinder sets just described. \hl{The advantage of working with cylinder sets (as will become evident shortly) is that unsafe cylinder sets are disjoint, making computations of probabilities much simpler.} }

\begin{customlem}{2}[Restated]\label{lemma:unsafe_sequences_escape_re}
Let $P_{\mu_a}(Q)$ be the probability of $Q\subset C_\infty$, under the assumption that the sequence of data is generated by i.i.d. Bernoulli random variables of parameter $\mu_a$.
Then the following statements hold:
\begin{enumerate}[label=\roman*)]
    \item Under $\mathcal{H}_1$ ($\mu_a \geq \mu$), i.i.d. sequences produced by such a distribution are correctly classified almost surely, that is
\begin{equation}
    P_{\mu_a}(Q_U)=1\,.
    \label{eq:p1Q1}
\end{equation}

\item Under $\mathcal{H}_0$ ($\mu_a\leq\mu-\epsilon$), the probability that a trajectory never rises above $\log(1/\alpha)$ is greater than $1-\alpha$, that is

$$
P_{\mu_a}(Q_S)\geq 1-\alpha\,.
$$

\end{enumerate}
\end{customlem}

\begin{proof}
We start by showing $\emph{i)}$.\\ For a given sequence $(d_1,\ldots,d_t,\cdots)$ we have
\begin{equation}
\Lambda_t = \log\prod_{i=1}^t\frac{f_\mu(d_i)}{f_{\mu-\epsilon}(d_i)}=\sum_{i=1}^t\log\frac{f_\mu(d_i)}{f_{\mu-\epsilon}(d_i)}\,.
\label{eq:log_likelihood}
\end{equation}
Dividing by $t$:
\begin{equation}
\frac{\Lambda_t}{t}=\frac{1}{t}\sum_{i=1}^t\log\frac{f_\mu(d_i)}{f_{\mu-\epsilon}(d_i)}\,.
\label{eq:lambda-sum}\end{equation}
By the Law of large numbers, as $t$ grows \eqref{eq:lambda-sum} converges to the expectation of the right hand side under $\mu_a$:
\begin{align}
\frac{\Lambda_t}{t}&\rightarrow \mathbb{E}_{\mu_a}\left[\log\frac{f_\mu(d)}{f_{\mu-\epsilon}(d)}\right]\nonumber\\
&=\mu_a \left(\log\frac{\mu}{\mu-\epsilon}-\log\frac{1-\mu}{1-(\mu-\epsilon)}\right)+\log\frac{1-\mu}{1-(\mu-\epsilon)}  \nonumber\\
&\geq\mu \left(\log\frac{\mu}{\mu-\epsilon}-\log\frac{1-\mu}{1-(\mu-\epsilon)}\right)+\log\frac{1-\mu}{1-(\mu-\epsilon)}  \nonumber\\
&=\text{kl}(\mu,\mu-\epsilon)>0\,,\label{eq:DKL} 
\end{align}
where  $\text{kl}$ is the Kullback-Leibler divergence between two Bernoulli distributions of parameters $\mu$ and $\mu-\epsilon$. The inequality holds because $\mu_a>\mu$ and the expression in brackets is positive. It follows that
$$
\lim_{t\rightarrow\infty}\Lambda_t = \infty,\quad a.s.
$$
Therefore there must exist a positive integer $t$ for which $\Lambda_t$ exceeds $\log (1/\alpha)$, so that the sequence $(d_\tau)_{\tau\geq 1}$ belongs to an  unsafe cylinder of order $t$ and thus $(d_\tau)_{\tau\geq 1}\in Q_U$, which is what we wanted to show.\qed
\\

Next, we prove \emph{ii)}, first, by showing that $P_{\mu-\epsilon}(Q_S)\geq 1-\alpha$ for the limiting case $\mu_a=\mu-\epsilon$, and then by generalizing it for $\mu_a\leq \mu-\epsilon$.
For $\mu_a=\mu-\epsilon$, the core of the proof relies in showing that when the threshold for declaring $\mathcal{H}_1$ is set to $\log(1/\alpha)$, then:
\begin{equation}
 P_\mu(Q_U) \geq \frac{1}{\alpha} P_{\mu-\epsilon} (Q_U)\; . \label{eq:pQ1}
\end{equation}

Once we prove \eqref{eq:pQ1}, we use the fact that $P_\mu(Q_U)=1$ (see  \eqref{eq:p1Q1}) and obtain:
\begin{equation}
P_{\mu-\epsilon}(Q_U)\leq \alpha \Rightarrow P_{\mu-\epsilon}(Q_S) \geq 1- \alpha\,, 
\label{eq:P0Qs}    
\end{equation}
as desired. To prove \eqref{eq:pQ1} we start by decomposing $Q_U$ as the union across time $t$ of the union of all unsafe cylinders of order $t$, that is:  $$
Q_U = \bigcup_{t=1}^\infty\bigcup_{(x_1,\ldots,x_t) \in \mathcal X_t} C(x_1,\ldots,x_t)\,,
$$
where $\mathcal X_t$ collects the tuples  $(x_1,\ldots,x_t)$ that define  unsafe cylinders of order $t$, i.e., those  satisfying \eqref{eq:define_cylinder1} and \eqref{eq:define_cylinder2}.  
By construction all the cylinder sets are disjoint, hence: 
\begin{align}
P_\mu(Q_U)&=\sum_{t=1}^\infty\sum_{(x_1,\ldots,x_t) \in \mathcal X_t} P_\mu\left(C(x_1,\ldots,x_t)\right)\label{eq:cyl-ineq1}\\
&=\sum_{t=1}^\infty\sum_{(x_1,\ldots,x_t) \in \mathcal X_t} f_\mu(x_1,\ldots,x_t)\label{eq:cyl-ineq2}\\
&\geq \sum_{t=1}^\infty\sum_{(x_1,\ldots,x_t) \in \mathcal X_t} \frac{1}{\alpha} f_{\mu-\epsilon}(x_1,\ldots,x_t)\label{eq:cyl-ineq3}\\
&=\frac{1}{\alpha}\sum_{t=1}^\infty\sum_{(x_1,\ldots,x_t) \in \mathcal X_t} P_{\mu-\epsilon}(C(x_1,\ldots,x_t))\label{eq:cyl-ineq4}\\
&=\frac{1}{\alpha}P_{\mu-\epsilon}(Q_U)\,,\label{eq:cyl-ineq5}
\end{align}
where the second identity follows from marginalizing over future trajectories (e.g.: in Fig. \ref{fig:sprt_schematic} the marginalization would be done over all the gray trajectories after $\Lambda_4$),
and the inequality holds since $\mathcal X_t$ is defined so that \eqref{eq:define_cylinder1} holds.

Now that we have \eqref{eq:pQ1},  \eqref{eq:P0Qs} follows immediately, and we move to the second part of the proof. We need to prove that $P_{\mu_a}(Q_S)\geq P_{\mu-\epsilon}(Q_S)$,  or equivalently $P_{\mu_a}(Q_U)\leq P_{\mu-\epsilon}(Q_U)$, when $\mu_a< \mu-\epsilon$. This essentially means that the probability of (incorrectly) classifying a safe machine as unsafe decreases as $\mu_a$ lowers, which is intuitively true.

To show this, notice that the log-likelihood ratio $\Lambda_t$ can be put in terms of $k$, the number of outcomes of $D_\tau=1$ over $t$ total pulls:

\begin{equation}
    \Lambda_t = k\log\frac{\mu}{\mu-\epsilon}+(t-k)\log\frac{1-\mu}{1-(\mu-\epsilon)}\,,\label{eq:lambda-as-k}
\end{equation}
and 
$k\sim\mathrm{Binomial}(t, \mu_a)$. 
What we want to show is that as $\mu_a$ lowers, the probability that $\Lambda_t\geq\log(1/\alpha)$ lowers as well, which is to be expected since $k$ will likely be lower. Assume then we declare $\mathcal{H}_1$, so departing from \eqref{eq:lambda-as-k} we can write:
$$\Lambda_t=k\lambda_0+(k-t)\lambda_1\geq\log(1/\alpha)\;,$$
where $\lambda_0=\log\frac{\mu}{\mu-\epsilon}$ and $\lambda_1=\log\frac{1-\mu+\epsilon}{1-\mu}$ are both positive. Since $\log(1/\alpha)$ is also positive we have
$k/t\geq\frac{\lambda_1}{\lambda_0+\lambda_1}=\left(1+\frac{\lambda_0}{\lambda_1}\right)^{-1}$. From the definition of $\lambda_0$ and $\lambda_1$ it yields:
\begin{align}
 \frac{k}{t}&\geq\left(1+\frac{\lambda_0}{\lambda_1}\right)^{-1}=\left(1+\frac{\log\frac{\mu}{\mu-\epsilon}}{\log\frac{1-\mu+\epsilon}{1-\mu}}\right)^{-1}\nonumber\\
 &\geq \left(1+\frac{\frac{\mu}{\mu-\epsilon}-1}{1-\frac{1-\mu}{1-\mu+\epsilon}}\right)^{-1}=\left(1+\frac{\frac{\epsilon}{\mu-\epsilon}}{\frac{\epsilon}{1-\mu+\epsilon}}\right)^{-1} \label{eq:log_ineq}\\
 &=\left(\frac{1}{\mu-\epsilon}\right)^{-1}=\mu-\epsilon > \mu_a\,, \label{eq:ktmua}
\end{align}
where the inequality in \eqref{eq:log_ineq} follows from the usual bounds of the logarithm $1-1/x\leq \log(x)\leq x-1$. Rearranging \eqref{eq:ktmua} we get $k-t\mu_a>0$. Then, the derivative of $f_{\mu_a}(d_1,\ldots,d_t)$ takes the form:
\begin{align}
&\frac{d}{d{\mu_a}} f_{\mu_a}(d_1,\ldots,d_t)=\frac{d}{d\mu_a}\binom{t}{k}\mu_a^k(1-\mu_a)^{t-k}\nonumber\\
&=\binom{t}{k} \left(k \mu_a^{k-1}(1-\mu_a)^{t-k}+\mu_a^k(-1)(t-k)(1-\mu_a)^{t-k-1}\right)\nonumber\\
&=\binom{t}{k}\mu_a^{k-1}(1-\mu_a)^{t-k-1}\left(k-t\mu_a\right)> 0\,,\label{eq:derivada}
\end{align}
where the last inequality stems from \eqref{eq:ktmua}. \hl{With \eqref{eq:derivada} we have that $f_{\mu_a}(\cdot) < f_{\mu-\epsilon}(\cdot)$ whenever we declare $\mathcal{H}_1$.
Then, going from \eqref{eq:cyl-ineq5} to \eqref{eq:cyl-ineq3} we can further lower bound the right-hand side of \eqref{eq:cyl-ineq3} by $\frac{1}{\alpha}f_{\mu_a}(x_1,\ldots,x_t)$ and arrive at $\frac{1}{\alpha}P_{\mu-\epsilon}(Q_U) \geq \frac{1}{\alpha}P_{\mu_a}(Q_U)$. Using the inequality in \eqref{eq:P0Qs} finishes the proof.
}
\end{proof}




\begin{customlem}{3} For a fixed arm $a$ of parameter $\mu_a$, consider the sequential probability ratio test defined by \eqref{eq:hyp-test}--\eqref{eq:log-likelihood-threshold}, where $\mu$, $\epsilon$ and $\alpha$ are given. Then, if the alternative $\mathcal{H}_1$ is true, the test is expected to terminate after $T$ steps, with \begin{equation}
    \mathbb{E}[T]\leq 1+\frac{\log\left(1/\alpha\right)}{\emph{kl}(\mu,\mu-\epsilon)}\;,
\end{equation}
where $\emph{kl}(\mu,\mu-\epsilon)$ is the Kullback-Leibler divergence between Bernoulli distributions
$$
\emph{kl}(\mu,\mu-\epsilon)=\mu\log\frac{\mu}{\mu-\epsilon}+\left(1-\mu\right)\log\frac{1-\mu}{1-\mu+\epsilon}\, .
$$
\end{customlem}
\begin{proof} 
Let $T$ be the smallest integer for which the test leads to the acceptance of $H_1$. Such variable is well defined and finite as a result of Lemma \ref{lemma:unsafe_sequences_escape}. Consider the sequence of damage up to time t $(d_t)_{t=1}^T$, and let $\mathbb{E}_{\mu_a}[\cdot]$ be the expectation with respect to the true distribution of the damage data (that is, $D\sim\text{Bernoulli}(\mu_a)$). Then
\begin{align}
    \mathbb{E}_{\mu_a}[\Lambda_T]&=\mathbb{E}_{\mu_a}\left[\sum_{t=1}^T \log\frac{f_\mu(d_t)}{f_{\mu-\epsilon}(d_t)}\right]\nonumber\\
    &=\mathbb{E}[T]\mathbb{E}_{\mu_a}\left[\log\frac{f_\mu(d)}{f_{\mu-\epsilon}(d)}\right]=\mathbb{E}[T]R_a\;,\label{eq:nbounded_1}
\end{align}
with $R_a=\mathbb{E}_{\mu_a}\left[\log\frac{f_\mu(d)}{f_{\mu-\epsilon}(d)}\right]$. Here we used Wald's identity \cite{wald-identity} in the second equality. Furthermore
%
%
%
%
\begin{align}
\mathbb{E}_{\mu_a}\left[\Lambda_T\right] &= \mathbb{E}_{\mu_a}\left[\Lambda_{T-1}+\log\frac{f_\mu(d_T)}{f_{\mu-\epsilon}(d_T)}\right]\nonumber\\
&= \mathbb{E}_{\mu_a}\left[\Lambda_{T-1}\right]+R_a\leq \log(1/\alpha) + R_a\,. \label{eq:nbounded_2}
\end{align}
Combining \eqref{eq:nbounded_1} and \eqref{eq:nbounded_2}:
$$
\mathbb{E}[T]\leq  1+\frac{\log\left(1/\alpha\right)}{R_a}\leq  1+\frac{\log\left(1/\alpha\right)}{\text{kl}(\mu,\mu-\epsilon)}\,,
$$
in virtue of $R_a \geq \text{kl}(\mu,\mu-\epsilon)$ for $\mu_a>\mu-\epsilon$, as was shown in \eqref{eq:DKL}. 
\end{proof}

}{} 

\subsection{Proof of Theorem \ref{thm:mab-relaxed-finite-stop-time}}\label{app:mab-relaxed-finite-time}
By lemma 3, we have
$$
    \expc[E_T]=\sum_{a=1}^M\expc[T_a]\leq {M}\left(1+\frac{\log(1/\alpha)}{\text{kl}(\mu,\mu-\epsilon)}\right)\,.
$$
Then by Wald's identity~\cite{wald-identity}, we have
$$\expc [E_T]=\expc\lhp \sum_{t=1}^T\expc[\mathds{1}\{\mu_{A_t}>\mu\}]\rhp\geq \expc\lhp T\frac{\lambda}{K-M+1}\rhp\,,$$
where the second inequality is due to the fact that before $T$, the probability of sampling an unsafe machine with a $\lambda$-soft strategy is at least $\lambda/(K-M+1)$. Recalling the upper bound for $\expc[E_T]$ in Theorem \ref{thm:mab-relaxed-finite-time}, one obtains
$$\mathbb{E}[T]\leq \frac{{M}(K-M+1)}{\lambda}\left(1+\frac{\log(1/\alpha)}{\text{kl}(\mu,\mu-\epsilon)}\right)\,.$$

\subsection{Proof of Theorem \ref{thm:sample_comp_barrier_learner}}\label{app:pf_sample_comp_barrier_learner}

We prove Theorem \ref{thm:sample_comp_barrier_learner} in three steps:
\subsubsection{Reformulation of Algorithm \ref{alg:barrier_learner}}
We reformulate Algorithm \ref{alg:barrier_learner} as in Algorithm \ref{alg:barrier_learner_re}.

\begin{algorithm}[!ht]
\KwData{Constrained Markov Decision Process $\mathcal{M}$}
Initialize $B^{(0)}(s,a)=0, \forall (s,a)\in\mathcal{S}\!\times\! \mathcal{A}$\\
\For{$\tau=0,1,\cdots$}{
    Draw $(s_\tau,a_\tau)\sim \mathrm{Unif}(\mathcal{S}\!\times\! \mathcal{A})$\\
    Sample transition $(s_\tau,a_\tau, s'_\tau,d_\tau)$ according to $\prob\lp S_1=s'_\tau,D_1=d_\tau|S_0=s_\tau,A_0=a_\tau\rp$\\
    \uIf{$B^{(\tau)}(s_\tau,a_\tau)\neq -\infty$}{$B^{(\tau+1)}\la\texttt{barrier\_update(}B^{(\tau)},s_\tau,a_\tau,s'_\tau,d_\tau\texttt{)}$}
    \Else{$B^{(\tau+1)}\la B^{(\tau)}$}
}
\caption{Barrier Learner Algorithm Reformulated}
\label{alg:barrier_learner_re}
\end{algorithm}

In the reformulated Algorithm \ref{alg:barrier_learner_re}, the sampling process is independent of $B$-function: At each iteration $\tau$, an $(s_\tau,a_\tau)$ pair is drawn uniformly from $\mathcal{S}\!\times\!\mathcal{A}$ and then a transition $(s_\tau,a_\tau,s'_\tau,d_\tau)$ is sampled according to the MDP, and the algorithm decides whether to accept such a sample depending on the value of $B(s_\tau,a_\tau)$. When we restrict ourselves to the trajectory of samples that are accepted, i.e. $$\{(s_\tau,a_\tau,s'_\tau,d_\tau): B^{(\tau)}(s_\tau,a_\tau)\neq -\infty,\ \tau=0,1,\cdots\}\,,$$
this trajectory is also a sampled trajectory of original Algorithm \ref{alg:barrier_learner}. More importantly, for such a trajectory, the probability it appears in original Algorithm \ref{alg:barrier_learner} is the same as the probability it appears as the accepted trajectory in Algorithm \ref{alg:barrier_learner_re}. With that, we define 
\be T_r:=\min\{\tau:B^{(\tau)}=B^*\}\,,\label{eq_tr_barrier_learner_re}\ee i.e. the earliest time when Algorithm \ref{alg:barrier_learner_re} detects all unsafe state-action pairs, then we have
\be
    \expc [T]=\expc\lhp \sum_{\tau=1}^{T_r}\mathds{1}\{B^{(\tau)}(s_\tau,a_\tau)\neq -\infty\}\rhp\,,\label{eq_sample_comp_re}
\ee
where $T$ is the earliest time when Algorithm \ref{alg:barrier_learner} detects all unsafe state-action pairs, as defined in Theorem \ref{thm:sample_comp_barrier_learner}. Expectations are taken with respect to the respective sampling processes of Algorithm \ref{alg:barrier_learner} and \ref{alg:barrier_learner_re}, which are different. With \eqref{eq_sample_comp_re}, it suffices to analyze the expected detection time of Algorithm \ref{alg:barrier_learner_re}. 

\subsubsection{Construction of modified algorithm}
As discussed in Section \ref{ssec:3.3}, an $(s,a)$ pair is unsafe if either it causes damage immediately or it transitions to an unsafe state with non-zero probability. If only the latter happens for such an unsafe $(s,a)$, then to be able to declare it unsafe, one must have already declared one of its succeeding states unsafe. To make such intuition precise, we recursively define disjoint subsets $\mathcal{S}_l,l=1,2,\cdots$ of the state space $\mathcal{S}$ as follow,
\begin{align}
    \mathcal{S}_1&:=\lb s\in\mathcal{S}: \prob_\pi\lp D_1=1|S_0=s\rp>0,\forall \pi\rb\,,\nonumber\\
    \mathcal{S}_l&:=\begin{cases}
        \mathcal{S}_l',& \mathcal{S}_l'\neq \emptyset\\
        \mathcal{S}\setminus{\bigcup_{k<l}\mathcal{S}_k}, &\mathcal{S}_l'=\emptyset
    \end{cases}\label{eq:def_unsafe_states_sets}\,,\\
    &\!\!\!\!\!\!\text{where}\nonumber\\
    \mathcal{S}'_l&=\lb s\in\mathcal{S}\setminus{\bigcup_{k<l}\mathcal{S}_k}:\prob_\pi\lp S_1\in\bigcup_{k<l}\mathcal{S}_k|S_0=s\rp>0,\forall \pi\rb\nonumber 
\end{align}

Observe that for any  finite MDP, its \emph{lag}  $L:=\max\{l>0:\mathcal{S}_{l+1}\neq\emptyset\}$ is finite. Following the definition \eqref{eq:def_unsafe_states_sets}, $\{\mathcal{S}_l,l=1,\cdots,L,L+1\}$ is a partition of $\mathcal{S}$. Any state $s_0\in\bigcup_{l=1}^L\mathcal{S}_l$ is unsafe because starting from $s_0$ and  under any policy $\pi$,  the MDP eventually reaches a state in $\mathcal{S}_1$ with non-zero probability, then causes damage. Furthermore, any state $s_0\in\mathcal{S}_{L+1}:=\mathcal{S}\setminus{\bigcup_{l=1}^L\mathcal{S}_l}$ is safe since $\mathcal{S}'_{L+1}=\emptyset$ implies that there exists $a_0\in\mathcal{A}$ such that $\prob( S_1\in\bigcup_{k<L+1}\mathcal{S}_k|S_0=s_0,A_0=a_0)=0$, i.e. taking action $a_0$ keeps the MDP away from the unsafe states in $\bigcup_{l=1}^L\mathcal{S}_l$. However, we note that a safe state can have unsafe actions and they can be detected by the Barrier Learner Algorithm. 
More importantly, the unsafe state sets $\mathcal{S}_l,l=1,2,\cdots,L$ satisfies that if all states in $\bigcup_{k<l}\mathcal{S}_k$ has been declared unsafe, any state-action pair in $\{(s,a):s\in\mathcal{S}_l,a\in\mathcal{A}\}$, when sampled, can be declared unsafe with non-zero probability. Base on this property, we construct an modified barrier learning algorithm using the prior information on $\mathcal{S}_l,l=1,2,\cdots,L$. 
The modified algorithm is described in Algorithm \ref{alg:barrier_learner_mod}.
\begin{algorithm}[!ht]
\KwData{Constrained Markov Decision Process $\mathcal{M}$, $\mathcal{S}_l,l=1,2,\cdots,L,L+1$ defined for $\mathcal{M}$}
Initialize $\hat{B}^{(0)}(s,a)=0, \forall (s,a)\in\mathcal{S}\!\times\! \mathcal{A}$\\
Initialize $l=1$\\
\For{$\tau=0,1,\cdots$}{
    Draw $(s_\tau,a_\tau)\sim \mathrm{Unif}(\mathcal{S}\!\times\! \mathcal{A})$\\
    Sample transition $(s_\tau,a_\tau, s'_\tau,d_\tau)$ according to $\prob\lp S_1=s'_\tau,D_1=d_\tau|S_0=s_\tau,A_0=a_\tau\rp$\\
    \uIf{$\hat{B}^{(\tau)}(s_\tau,a_\tau)\neq -\infty$ and $s_\tau\in\mathcal{S}_{l}$}{$\hat{B}^{(\tau+1)}\leftarrow\texttt{barrier\_update(}\hat{B}^{(\tau)},s_\tau,a_\tau,s'_\tau,d_\tau\texttt{)}$}
    \Else{$\hat{B}^{(\tau+1)}\la \hat{B}^{(\tau)}$}
    \If{$\hat{B}^{\tau+1}(s,a)=-\infty,\forall s\in\mathcal{S}_{l},a\in\mathcal{A}$}{$l\la l+1$}
}
\caption{Modified Barrier Learner Algorithm with Prior Information on $\mathcal{S}_l,l=1,2,\cdots,L,L+1$}
\label{alg:barrier_learner_mod}
\end{algorithm}

The modified algorithm is similar to Algorithm \ref{alg:barrier_learner_re} but it learns $\mathcal{S}_l,l=1,2,\cdots,L$ in order: At the beginning ($l=1$), it only declares $(s,a)$ pairs associated with $\mathcal{S}_1$ unsafe until all states in $\mathcal{S}_1$ are declared unsafe, after which $l$ increases to $2$. Now the algorithm only declares $(s,a)$ pairs associated with $\mathcal{S}_2$ unsafe. Finally after all states in $\bigcup_{l=1}^L\mathcal{S}_l$ are declared unsafe ($l=L+1$), the algorithm starts to learn the unsafe transitions for safe states in $\mathcal{S}_{L+1}$. Similarly, we define 
\be \hat{T}_r:=\min\{\tau:\hat{B}^{(\tau)}=B^*\}\,,\label{eq_tr_barrier_learner_mod}\ee i.e. the earliest time when Algorithm \ref{alg:barrier_learner_mod} detects all unsafe state-action pairs.
Since the modified algorithm is more restrictive on declaring unsafe state-action pair, the expected detection time of the modified algorithm is no less than that of Algorithm \ref{alg:barrier_learner_re}, as stated in the following claim.

\begin{claim}\label{clm:s_order_of_b_learner}
    Given an MDP, let $T_r$ and $\hat{T}_r$ be the earliest times when Algorithm \ref{alg:barrier_learner_re} and Algorithm \ref{alg:barrier_learner_mod}, detect all unsafe state-action pairs in this MDP, respectively, as defined in \eqref{eq_tr_barrier_learner_re} and \eqref{eq_tr_barrier_learner_mod}.
    Then, $\expc\lhp \sum_{\tau=1}^{\hat{T}_r}\mathds{1}\{\hat{B}^{(\tau)}(s_\tau,a_\tau)\neq -\infty\}\rhp$
 is lower-bounded by $\expc\lhp \sum_{\tau=1}^{T_r}\mathds{1}\{B^{(\tau)}(s_\tau,a_\tau)\neq -\infty\}\rhp\qquad\qquad$, where expectations are w.r.t.  $\{(s_\tau,a_\tau,s'_\tau,d_\tau),\tau=0,1,\cdots\}$.
    
\end{claim}
\begin{proof}
    Condition on a fixed sample trajectory $\mathcal{T}:=\{(s_\tau,a_\tau,s'_\tau,d_\tau)\}_{\tau=0}^\infty\,,$ the functions $B^{(\tau)}$ and $\hat{B}^{(\tau)}$ are deterministic. We have
    \begin{align}
        B^{(\tau)}(s,a)|\mathcal{T}\leq \hat{B}^{(\tau)}(s,a)|\mathcal{T},\forall \tau\!\geq\!0, \forall (s,a)\!\in\!\mathcal{S\!\times\! \mathcal{A}}\,,\label{eq_unsafesets_inclusion}
    \end{align}
    proved by induction: we have $$B^{(0)}(s,a)|\mathcal{T}\leq \hat{B}^{(0)}(s,a)|\mathcal{T}, \forall (s,a)\in\mathcal{S\!\times\! \mathcal{A}}\,,$$
    at initialization. Suppose that \eqref{eq_unsafesets_inclusion} holds at time $\tau=t$. If $(s_t,a_t,s'_t,d_t)$ is accepted by both algorithms, or rejected by both algorithms, we have 
    \be B^{(t+1)}(s,a)|\mathcal{T}\leq \hat{B}^{(t+1)}(s,a)|\mathcal{T},\forall (s,a)\in\mathcal{S\!\times\! \mathcal{A}}\,.\label{eq_unsafesets_inclusion_ind}\ee
    
    If $(s_t,a_t,s'_t,d_t)$ is rejected by Algorithm \ref{alg:barrier_learner_re} and accepted by Algorithm \ref{alg:barrier_learner_mod}, then we have
    $
        B^{(t)}(s_t,a_t)=-\infty,\hat{B}^{(t)}(s_t,a_t)=0\,.
    $
    \eqref{eq_unsafesets_inclusion_ind} still holds, since only $\hat{B}^{(t+1)}(s_t,a_t)$ is updated to either $0$ or $-\infty$. If $(s_t,a_t,s'_t,d_t)$ is accepted by Algorithm \ref{alg:barrier_learner_re} and rejected by Algorithm \ref{alg:barrier_learner_mod}, then we have
    $$
        B^{(t)}(s_t,a_t)=\hat{B}^{(t)}(s_t,a_t)=0\,.
    $$
    Inequality \eqref{eq_unsafesets_inclusion_ind} still holds, since only $B^{(t+1)}(s_t,a_t)$ is updated to either $0$ or $-\infty$. Now from \eqref{eq_unsafesets_inclusion}, we immediately know that condition on the fixed sample trajectory $\mathcal{T}$,
    $$
        \mathds{1}\{B(s_\tau,a_\tau)\neq -\infty\}\leq \mathds{1}\{\hat{B}(s_\tau,a_\tau)\neq -\infty\},\forall \tau=0,1,\cdots
    $$
    Notice that $T_r$ ($\hat{T}_r$) is the minimum $t$ such that $B^{(t)}$ ($\hat{B}^{(t)}$) becomes exactly the same as $B^*$. Then
    $
        T_r|\mathcal{T}\leq \hat{T}_r|\mathcal{T}\,.
    $
    Therefore one have, by law of total expectation,
    \begin{align*}
        &\;\expc\lhp \sum_{\tau=0}^{T_r}\mathds{1}\{B(s_\tau,a_\tau)\neq -\infty\}\rhp\\
        =&\;\expc\lhp\expc\lhp \left.\sum_{\tau=0}^{T_r}\mathds{1}\{B(s_\tau,a_\tau)\neq -\infty\}\rv\mathcal{T}\rhp\rhp\\
        \leq&\; \expc\lhp\expc\lhp \left.\sum_{\tau=0}^{\hat{T}_r}\mathds{1}\{\hat{B}(s_\tau,a_\tau)\neq -\infty\}\rv\mathcal{T}\rhp\rhp\\
        =&\;\expc\lhp \sum_{\tau=0}^{\hat{T}_r}\mathds{1}\{\hat{B}(s_\tau,a_\tau)\neq -\infty\}\rhp
    \end{align*}
    \vspace{-10pt}
\end{proof}
\subsubsection{Expected detection time of modified algorithm}
Lastly, we prove the following Theorem regarding the expected detection time of the modified algorithm.
\begin{theorem}\label{thm:sample_comp_barrier_learner_mod}
    Given an MDP with $\mathcal{S}_l,l=1,2,\cdots,L,L+1$ defined as in \eqref{eq:def_unsafe_states_sets}. Assume that exists $\rho>0$ such that the transition probability $\prob(S_1=s'|S_0=s,A_0=a)$, is either zero or lower bounded by $\rho$, for all $s,s'\in\mathcal{S},a\in\mathcal{A}$. 
    Let $\hat{T}_r$ be earliest time when Algorithm \ref{alg:barrier_learner_mod} detects all unsafe state-action pairs as defined in \eqref{eq_tr_barrier_learner_mod}, then we have
    $$
        \expc\lhp \sum_{\tau=1}^{\hat{T}_r}\mathds{1}\{\hat{B}^{(\tau)}(s_\tau,a_\tau)\neq -\infty\}\rhp\leq \frac{|\mathcal{S}||\mathcal{A}|}{\rho}\sum_{l=1}^{L+1}\lp\sum_{k=1}^{|\mathcal{S}_l||\mathcal{A}|}\frac{1}{k}\rp\,.
    $$
\end{theorem}
\begin{proof}
    Let $\hat{T}_l,l=1,\cdots,L+1$ denote the earliest time when all unsafe state-action pairs associated with $\mathcal{S}_l$ are detected by Algorithm \ref{alg:barrier_learner_mod}, and we let $\hat{T}_0=0$. Then clearly $$\hat{T}_{l-1}< \hat{T}_l,l=1,\cdots,L+1\,,\quad \hat{T}_{L+1}=\hat{T}_r\,,$$
    and
    $
        \Delta_l:=\sum_{t=\hat{T}_{l-1}}^{\hat{T}_l-1} \mathds{1}\{\hat{B}^{(t)}(s_t,a_t)\neq -\infty\}
    $
    is the number of accepted samples by Algorithm \ref{alg:barrier_learner_mod} between $\hat{T}_{l-1}$ and $\hat{T}_l$.
    
    Notice that we can view the barrier learning process between $\hat{T}_{l-1}$ and $\hat{T}_l$ as detecting unsafe machines in the safe multi-arm bandits problem discussed in Section \ref{sec:multi_arm}: At time $\hat{T}_{l-1}$, there are in total $K_l:=\lv\lb (s,a):s\in\bigcup_{k=l}^{L+1}\mathcal{S}_k,a\in\mathcal{A}\rb\rv$ machines, and the number of unsafe machines is $$M_l:=\lv\{(s,a):s\in\mathcal{S}_l,a\in\mathcal{A}, B^*(s,a)=-\infty \}\rv\,.$$ We have
    $
        K_l\leq |\mathcal{S}||\mathcal{A}|\,,\ M_l\leq |\mathcal{S}_l||\mathcal{A}|\,.
    $
    Furthermore, condition on such an unsafe machine is pulled, i.e. an unsafe $(s,a)$ in $\mathcal{S}_l$ is accepted by Algorithm \ref{alg:barrier_learner_mod}, the probability of declaring it unsafe is at least $\rho$. Because  the $(s,a)$ pair either transitions to some $s\in\bigcup_{k=1}^{l-1}\mathcal{S}_k$ that has been declared unsafe or directly incurs damage with non-zero probability, and that probability is lower bounded by $\rho$ according to our assumption. 
    
    The acceptance of sample $(s_\tau,a_\tau,s'_\tau,d_\tau)$ is equivalent to pulling a uniformly randomly drawn arm out of arms that have not been declared unsafe. Theorem \ref{thm:mab-assured-finite-time} suggests that the expected number of such "pulling" is upper bounded as
    $$
        \expc [\Delta_l]\leq \frac{K_l}{\rho}\lp\sum_{k=1}^{M_l}\ \frac{1}{k}\rp\leq \frac{|\mathcal{S}||\mathcal{A}|}{\rho}\lp\sum_{k=1}^{|\mathcal{S}_l||\mathcal{A}|}\frac{1}{k}\rp\,.
    $$
    Finally, we have
    \begin{align*}
        &\;\expc\lhp \sum_{t=0}^{\hat{T}_r}\mathds{1}\{\hat{B}^{(t)}(s_t,a_t)\neq -\infty\}\rhp\\
        =&\;\expc \lhp \sum_{l=1}^{L+1}\Delta_l\rhp\leq \frac{|\mathcal{S}||\mathcal{A}|}{\rho}\sum_{l=1}^{L+1}\lp\sum_{k=1}^{|\mathcal{S}_l||\mathcal{A}|}\frac{1}{k}\rp\,.
    \end{align*}
\vspace{-6pt}
\end{proof}
\begin{proof}[Proof of Theorem \ref{thm:sample_comp_barrier_learner}]
Given any MDP, we have $|\mathcal{S}_l|\leq |\mathcal{S}|,\forall l=0,1,\cdots,L+1$, 
then we have
\begin{align*}
    \expc [T] &\leq\!\expc\lhp \sum_{\tau=1}^{\hat{T}_r}\mathds{1}\{\hat{B}^{(\tau)}(s_\tau,a_\tau)\neq -\infty\}\rhp\\
    &\leq\! \frac{|\mathcal{S}||\mathcal{A}|}{\rho}\sum_{l=1}^{L+1}\lp\sum_{k=1}^{|\mathcal{S}_l||\mathcal{A}|}\frac{1}{k}\rp\leq\! \frac{|\mathcal{S}||\mathcal{A}|(L+1)}{\mu}\lp\sum_{k=1}^{|\mathcal{S}||\mathcal{A}|}\frac{1}{k}\rp\,,
\end{align*}
where the first equality is from \eqref{eq_sample_comp_re} and Claim \ref{clm:s_order_of_b_learner}. \hl{The result follows by upper bounding the summation with the $\log$ inequality $\sum_{k=1}^n\frac{1}{k}\leq\log(n+1)$.}
\end{proof}

\bibliography{refs}
\bibliographystyle{ieeetr}
\vspace{-5ex}
\begin{IEEEbiography}[{\includegraphics[width=1in,height=1.25in,clip,keepaspectratio]{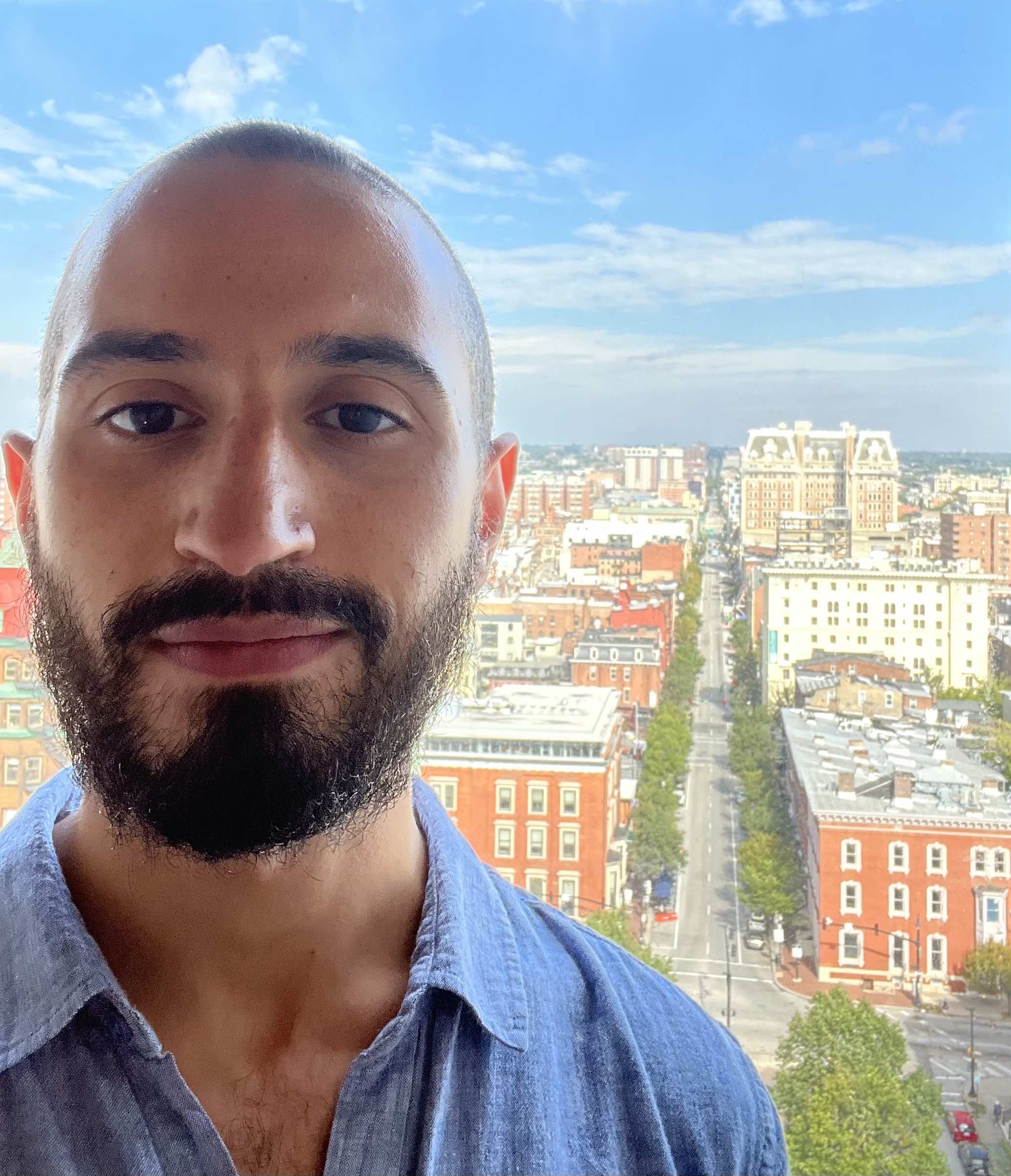}}]{Agustin Castellano} is pursuing a Ph.D. at the Department of Electrical Engineering at Johns Hopkins University. He
received the M.Sc. and his degree in Electrical Engineering from Universidad de la Rep\'ublica, Uruguay in 2021 and 2017 respectively. For his M.Sc. dissertation he was awarded the first prize given by the National Academy of Engineers. His current research interests include Reinforcement Learning theory and algorithms, with applications to power system optimization.
\end{IEEEbiography}
\vspace{-5ex}

\begin{IEEEbiography}
[{\includegraphics[width=1in,height=1.25in,clip,keepaspectratio]{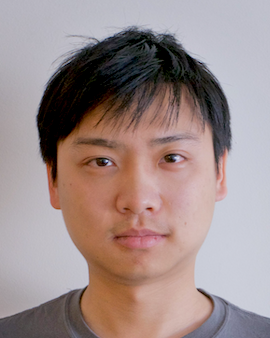}}]{Hancheng Min}
is currently working toward the Ph.D.
degree at the Department of Electrical and Computer Engineering, Johns Hopkins University. He received the B.Eng. degree in Electrical Engineering and Automation from Tongji University in 2016, and the M.S. degree in Systems Engineering from University of Pennsylvania in 2018. His research interests include analysis and control of large-scale networks, reinforcement learning and deep learning theory.
\end{IEEEbiography}
\vspace{-5ex}
\begin{IEEEbiography}[{\includegraphics[width=1in,height=1.25in,clip,keepaspectratio]{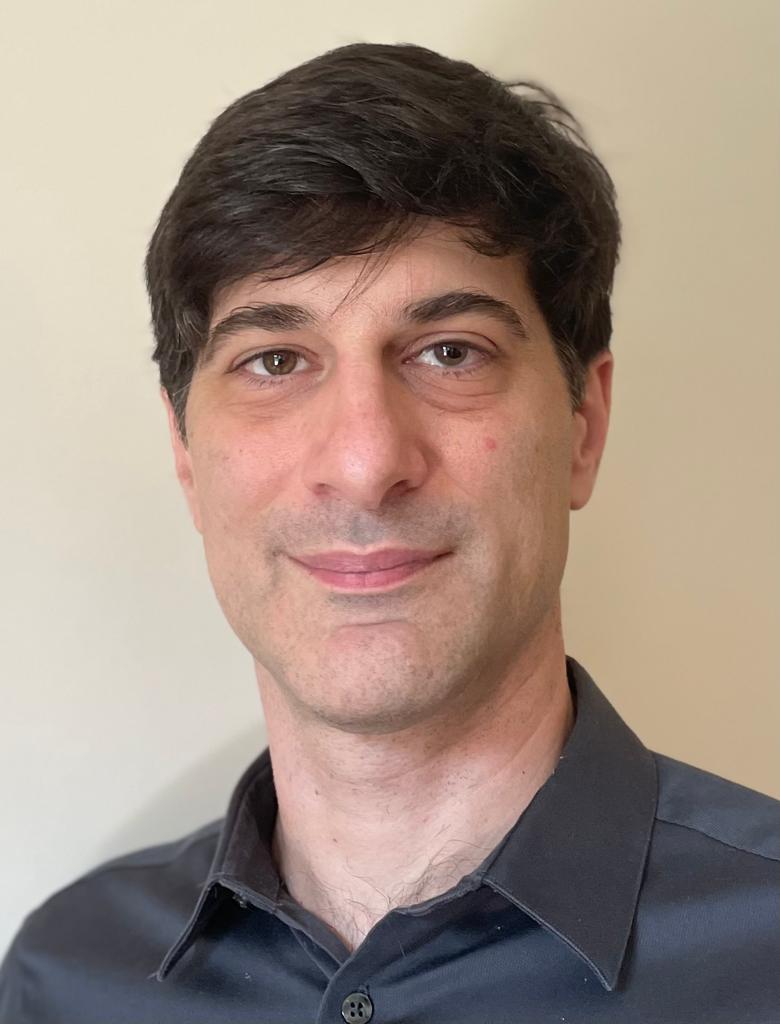}}]{Juan Andr\'es Bazerque} received the B.Sc. degree in electrical engineering from Universidad de la Republica (UdelaR), Montevideo, Uruguay, in 2003, ´
and the M.Sc. and Ph.D. degrees from the Department of Electrical and Computer Engineering, University of Minnesota (UofM), Minneapolis, in 2010 and 2013 respectively. 
After his PhD studies he entered the Department of Electrical Engineering at UdelaR as an Assistant Professor. In 2022 he moved back to the US where he joined the Department of Electrical and Computer Engineering at the University of Pittsburgh.  
His current research interests include stochastic optimization and networked systems, focusing on reinforcement learning, graph signal processing, and power systems optimization and control. 
Dr. Bazerque is the recipient of the UofM’s Master Thesis Award 2009-2010, and co-recipient of the best paper award at the 2nd International Conference on Cognitive Radio Oriented Wireless Networks and Communication 2007.
 \end{IEEEbiography}
\vspace{-5ex}

\begin{IEEEbiography}[{\includegraphics[width=1in,height=1.25in,clip,keepaspectratio]{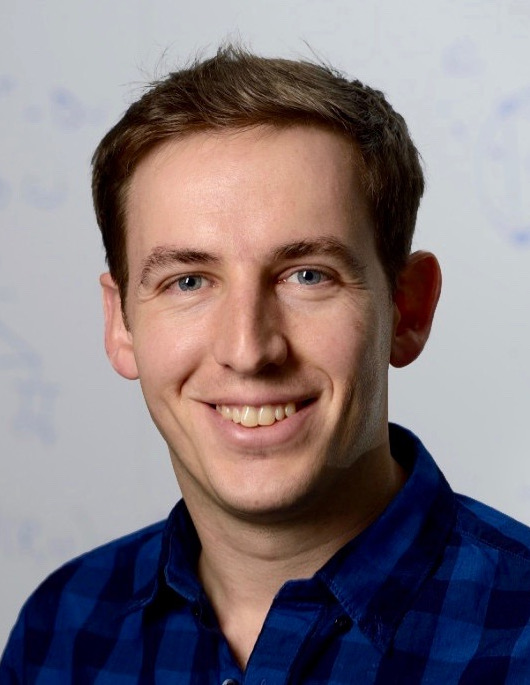}}]
{Enrique Mallada} (S'09-M'13-SM') is an Associate Professor of Electrical and Computer Engineering at Johns Hopkins University since 2022. Prior to joining Hopkins in 2016, he was a Post-Doctoral Fellow in the Center for the Mathematics of Information at Caltech from 2014 to 2016. He received his Ingeniero en Telecomunicaciones degree from Universidad ORT, Uruguay, in 2005 and his Ph.D. degree in Electrical and Computer Engineering with a minor in Applied Mathematics from Cornell University in 2014.
Dr. Mallada was awarded
the Johns Hopkins Alumni Association Teaching Award in 2021,
the Catalyst and Discovery Awards in 2020 and 2021, respectively, from Johns Hopkins University,
the NSF CAREER award in 2018,
the ECE Director's Ph.D. Thesis Research Award for his dissertation in 2014,
the Center for the Mathematics of Information (CMI) Fellowship from Caltech in 2014,
and the Cornell University Jacobs Fellowship in 2011.
His research interests lie in the areas of control, dynamical systems, and optimization, with applications to engineering networks.
\end{IEEEbiography}

\end{document}